\documentclass[letterpaper, 10 pt, conference]{ieeeconf}  

\IEEEoverridecommandlockouts                              

\usepackage{amsthm}
\usepackage{amsmath,amssymb}
\usepackage{graphicx}
\usepackage{latexsym,url}
\usepackage{graphicx}
\usepackage{comment}
\usepackage[outdir=./]{epstopdf}
\usepackage{stfloats}
\usepackage{color}
\usepackage{xcolor}
\usepackage{comment}
\usepackage{bm}

\makeatletter
\let\MYcaption\@makecaption
\makeatother

\usepackage{subcaption}

\makeatletter
\let\@makecaption\MYcaption
\makeatother

\newtheorem{theorem}{Theorem}

\newtheorem{definition}{Definition}
\newtheorem{corollary}{Corollary}

\newcommand{\SO}{SO}

\newcommand{\mc}{\mathcal}

\newcommand{\R}{{\mathbb R}}

\newcommand{\diag}{{\rm diag}}

\newcommand{\argmin}{\mathop{\rm arg~min}\limits}

\newcommand{\red}[1]{\textcolor{red}{#1}}
\newcommand{\blue}[1]{\textcolor{blue}{#1}}

\title{\LARGE \bf
Collision Avoidance for Elliptical Agents with Control Barrier Function Utilizing Supporting Lines
}

\author{Koju Nishimoto$^{1}$, Riku Funada$^{1}$, Tatsuya Ibuki$^{2}$, and Mitsuji Sampei$^{1}$
\thanks{*This work was supported by JSPS KAKENHI Grant Number 21K20425. To be presented at the 2022 American Control Conference.}
\thanks{$^{1}$ Department of Systems and Control Engineering, Tokyo Institute of Technology
{\tt\small  nishimoto.k.ab@m.titech.ac.jp.} {\tt\small\{funada\},\{sampei\}@sc.e.titech.ac.jp.}
}%
\thanks{$^{2}$ Department of Electronics and Bioinformatics, Meiji University
        {\tt\small ibuki@meiji.ac.jp.}}%
}

\begin{document}

\maketitle
\thispagestyle{empty}
\pagestyle{empty}

\begin{abstract}
This paper presents a collision avoidance method for elliptical agents traveling in a two-dimensional space. 
We first formulate a separation condition for two elliptical agents utilizing a signed distance from a supporting line of an agent to the other agent, which renders a positive value if two ellipses are separated by the line. 
Because this signed distance could yield a shorter length than the actual distance between two ellipses, the supporting line is rotated so that the signed distance from the line to the other ellipse is maximized. 
We prove that this maximization problem renders the signed distance equivalent to the actual distance between two ellipses, hence not causing the conservative evasive motion. 
Then, we propose the collision avoidance method utilizing novel control barrier functions incorporating a gradient-based update law of a supporting line. 
The validity of the proposed methods is evaluated in the simulations.
\end{abstract}

\section{INTRODUCTION} \label{sec:Intro}
Guaranteeing collision avoidance in multi-agent systems is of significant importance to ensuring safety in many application fields, including environmental monitoring \cite{Arslan16, Funada20}, autonomous transportation \cite{Miyano20, Gong17}, robot navigation \cite{Barbosa20}, and precision agriculture \cite{Zhang12}. In these challenging and complex realms, multi-agent systems are demanded to embrace agents of different capabilities, shapes, and sizes to enhance performance \cite{hetero}. In the presence of such heterogeneity, the collision avoidance protocol should guide agents not to collide with each other while incorporating their forms.

To achieve real-time collision avoidance for multi-agent systems, various approaches have been developed, including the methods utilizing artificial potential fields (APFs) \cite{Hoy2015} and control barrier functions (CBFs) \cite{Ames2019_CBF_thapp}. 
APFs were first presented in \cite{Khatib1986}, and have been employed to the multi-agent systems in the context of the formation \cite{Dimarogonas08} and flocking control \cite{Olfari-saber06, Gazi04}, where a repulsive potential function is designed for steering agents not to collide with each other.
A local repulsive function activated only in the sensing regions of each agent is proposed in \cite{Stipanovic07}.
%
On the other hand, CBFs were recently proposed in \cite{Ames2017}, which confines the state of the system in the set defined by the CBFs.
The work \cite{Wang2017} developed the collision avoidance methods for multi robot systems, which can be implemented in a distributed fashion.
The authors in \cite{Glotfelter19} developed the hybrid CBFs to achieve the collision avoidance of the agents with limited sensing ranges.
The comparative study of APFs and CBFs in obstacle avoidance scenarios is conducted in \cite{Singletary2020}. 
Most of the mentioned papers assume the agent as a single point, a circular disk, or a sphere. 
Although these methods can be employed to any agents by overestimating the original shape of agents to a sphere enclosing them, this approach could render too conservative evasive behavior if they have a nonspherical, especially elongated, body.

To mitigate this conservativeness, we propose the collision avoidance method capable of embracing the agents with heterogeneous elliptical shapes, as shown in Fig.~\ref{fig:problem_situation}.
Because the CBF provides the ability to synthesize the number of safety-critical constraints \cite{Glotfelter21}, e.g., the constraints ensuring the battery of the robot never depletes \cite{Gennaro20}, together with the collision avoidance, we opt for a CBF-based approach.
While, in the context of APFs, the work \cite{Do2013} proposes the flocking control for ellipsoidal agents with achieving collision avoidance, the condition used for repulsive potential functions becomes complicated and is not straightforwardly extendable to CBFs.
The work \cite{Dimarogonas2019} employs the result in the computer graphics field to develop the separation condition of elliptical agents. 
However, the physical interpretation of the metric utilized in the collision avoidance law is not readily understandable since it does not provide the distance between agents.
The extent-compatible CBF is developed in \cite{Srinivasan21}, where it can prevent the collision among agents having volume by solving a sum-of-squares optimization (SOS) program. Although this method can be applied to elliptical agents, the computational burden stemming from SOS programs might prevent implementations to high order systems. The work also proposes the sampling-based methods while the designer needs to assume the bounded control input.
The complexities and limitations in the existing collision avoidance methods for elliptical agents partially come from the difficulties in measuring the distance between two separated ellipses.
Although the computational methods deriving numerical solutions of the distance between two ellipses are proposed in 
the computer graphics field \cite{Wang2001,Choi2020}, the analytical solution of the distance is difficult to derive in a simple form.

To alleviate the difficulties of evaluating a distance between two ellipses, we propose a novel CBF that incorporates a signed distance from a supporting line of an elliptical agent to the other agent, as illustrated in Fig.~\ref{fig:ellipse_setting}. Because a naive selection of the supporting line could yield a shorter length than the actual distance between two ellipses, we propose a gradient ascent-based update law, where the supporting line is rotated on the boundary so that the distance between the line and the other ellipse is maximized. We then prove that the maximum value derived from this optimization problem is equivalent to the actual distance between two ellipses. A novel CBF incorporating the gradient ascent input to rotate the supporting line is designed. In addition, we prove that the proposed CBF is a valid one, namely, there always exists the control input to make the collision-free set forward invariance. Numerical simulations demonstrate that the proposed method achieves the collision avoidance between elliptical agents without exhibiting conservative evasive motions.

\begin{figure} 
    \centering
    \includegraphics[width=40mm]{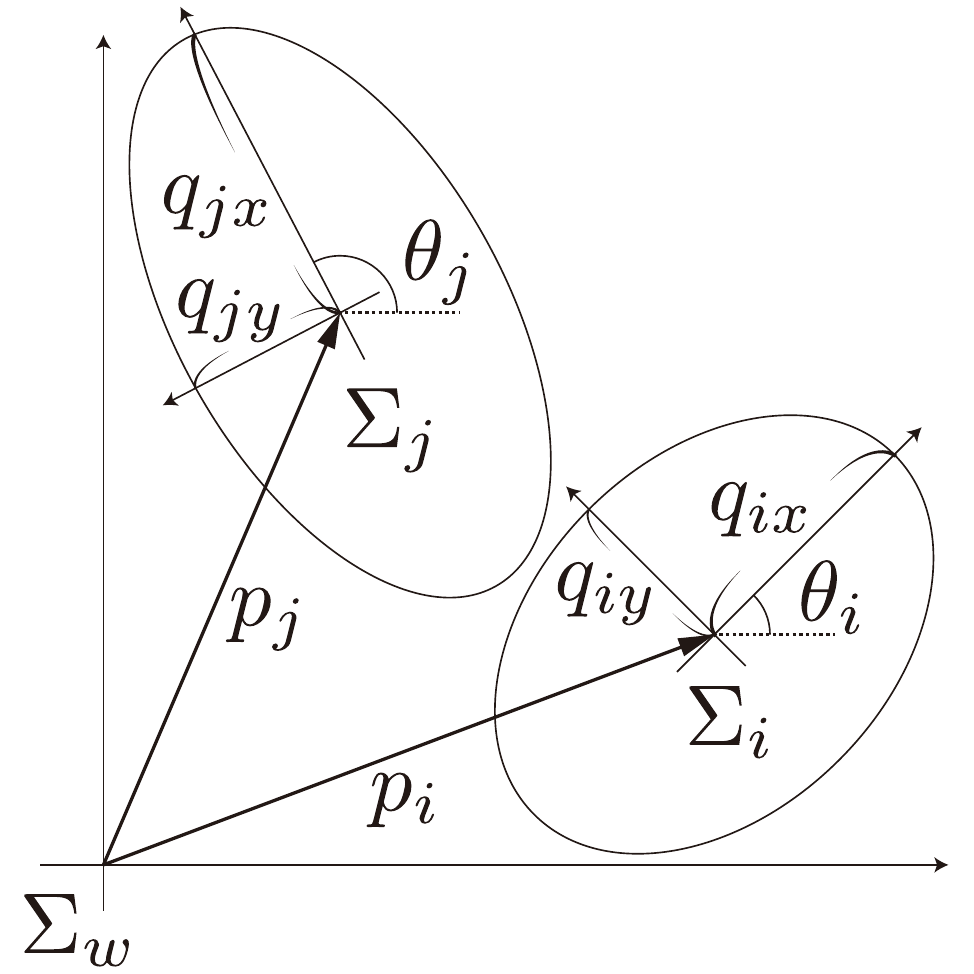}
    \caption{Proposed scenario. The agents characterized as ellipses with heterogeneous shapes avoid collisions with each other.}
    \label{fig:problem_situation}
\end{figure}

\section{Preliminary}
\subsection{Problem Formulation}
In this paper, we present a collision avoidance method among elliptical agents, labeled through the index set $\mc N = \{1 \cdots n\}$, in 2-D Euclidean space $\R^2$, as illustrated in Fig.~1. 
We denote the world coordinate frame as $\Sigma_w$.
We also define the coordinate frame of agent $i$ as $\Sigma_i$, arranged at the center of agent $i$ so that its $x_i$-axis corresponds with the major axis of the ellipse. The relative pose of $\Sigma_i$ with respect to $\Sigma_w$ is described as $(\bm{p}_i, R_i(\theta)): \R^2 \times \SO(2)$ with the position $\bm{p}_i=[p_{ix}, p_{iy}]^T \in \R^2$ and the orientation
\begin{align}
R_i(\theta_i) = 
\begin{bmatrix}
    \cos{(\theta_i)} & -\sin{(\theta_i)} \\
    \sin{(\theta_i)} & \cos{(\theta_i)}
\end{bmatrix},~ \theta_i \in (-\pi, \pi].
\end{align}

The state of agent $i$ is defined as $\bm{x}_i=[p_{ix},p_{iy},\theta_i]^T$. 
We suppose that the motion of agent $i$ can be represented according to a single integrator dynamics,
\begin{align} \label{eq:dynamics}
    \dot{\bm{x}}_{i}=\left[u_{ix},u_{iy},u_{i\theta}\right]^T
\end{align}
with the velocity input $[u_{ix},u_{iy}]^T$ and the angular velocity input $u_{i\theta}$. We denote the control input for agent $i$ as $\bm{u}_i = [u_{ix}, u_{iy}, u_{i\theta}]^T$.
Agent $i$ occupies the elliptical region $\mathcal{E}_i$ described as
\begin{align} \label{eq:ellipse}
    \mathcal{E}_i\!=\!\left\{\bm{X}\!\in\! \mathbb{R}^2 \!\mid\! (\bm{X}\!-\!\bm{p}_i)^TR_iQ_i^{-2}R_i^T(\bm{X}\!-\!\bm{p}_i)\!-\!1\!\leq\!0\right\},
\end{align}
where the constant matrix $Q_i=\diag(q_{ix}, q_{iy})$ is defined with $q_{ix}$ and $q_{iy}$ which specify the length of the major axis and the minor axis.


We propose a control method that prevents a collision between elliptical agents described with \eqref{eq:ellipse}. 
If the minimum distance between $\mathcal{E}_i$ and $\mathcal{E}_j$ is described as $w_{ij}^*(\bm{x}_{i},\bm{x}_{j})$, 
the safe set restricting collisions between agents $i$ and $j$ can be captured by the set
\begin{align}\label{eq:safeset}
    \mathcal{S}_{ij}=\left\{\bm{x}_{i},\bm{x}_{j}\in \mathbb{R}^3 \mid w_{ij}^*(\bm{x}_{i},\bm{x}_{j})\geq 0\right\},
\end{align}
where this set has to be rendered forward invariant. 
For this goal, we leverage control barrier functions (CBFs) being introduced in the next subsection.

\subsection{Control Barrier Functions}
\begin{figure}
    \centering
    \begin{minipage}[b]{0.45\hsize}
        \centering
        \includegraphics[width=25mm]{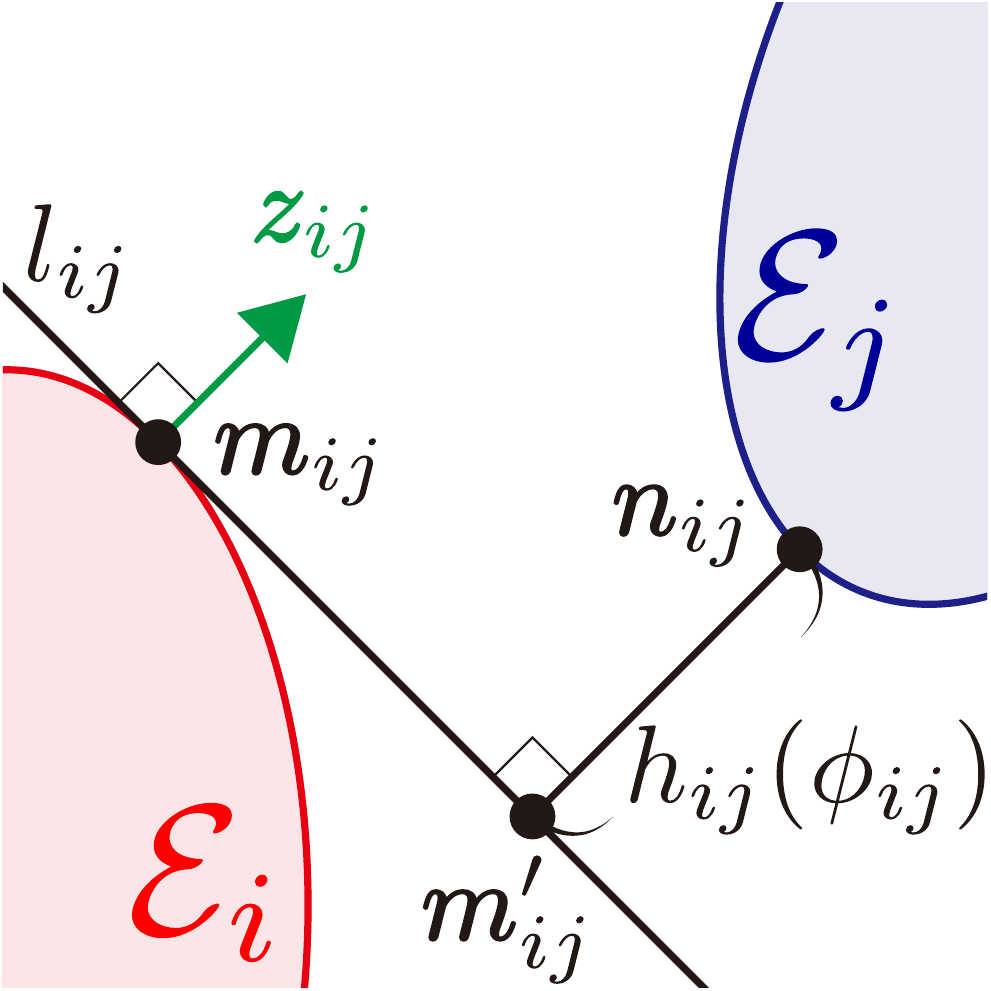}
        \subcaption{}\label{subfig:ellipse_setting_1}
    \end{minipage}
    \begin{minipage}[b]{0.45\hsize}
        \centering
        \includegraphics[width=25mm]{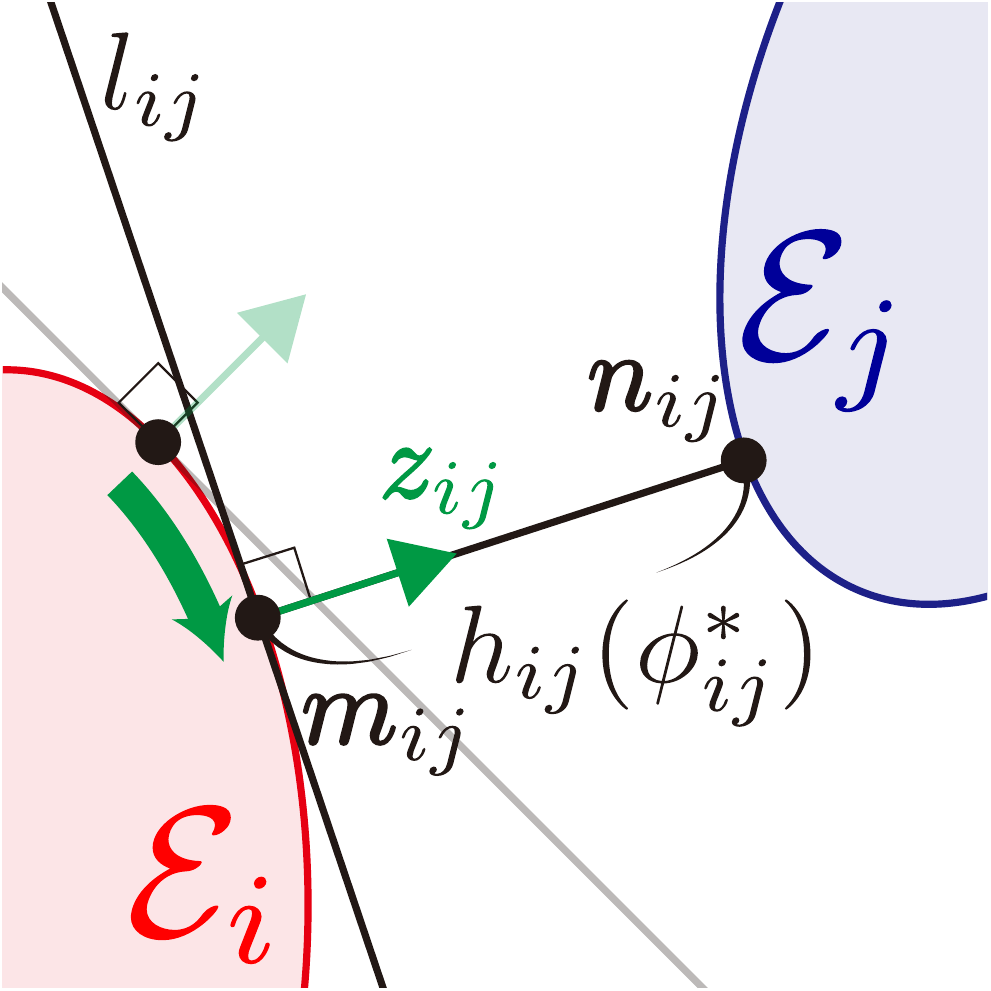}
        \subcaption{}\label{subfig:ellipse_setting_2}
    \end{minipage}
    \caption{The supporting line $l_{ij}$ separating two elliptical agents $\mc E_i$ and $\mc E_j$. 
    (a) shows the distance $h_{ij}(\phi_{ij})$ between the ellipse $\mc E_j$ and a supporting line $l_{ij}$, the normal vector and the tangent point of which is denoted as $\bm{z}_{ij}$ and $\bm{m}_{ij}$, respectively. 
    Although the distance $h_{ij}(\phi_{ij})$ can be derived from \eqref{ThisCBF}, $h_{ij}(\phi_{ij})$ could be shorter than the actual distance between two ellipses.
    (b) illustrates the update law of the supporting line $l_{ij}$, where $l_{ij}$ is rotated on the boundary of $\mc E_i$ so that $h_{ij}(\phi_{ij})$ approaches the actual distance between two ellipses, by maximizing $h_{ij}(\phi_{ij})$ with $\phi_{ij}$. Note that $h_{ij}(\phi_{ij})$ takes a positive value if and only if the line $l_{ij}$ separates two ellipses.}
    \label{fig:ellipse_setting}
\end{figure}
CBFs have been utilized for ensuring the forward invariance property to the set $\mc S$, in which the state $\bm{x}$ should be confined during the task execution of agents.
We assume that a set $\mc S$ can be expressed as the superlevel set of a continuously differentiable function $h: \R^n \to \R$, namely, 
${\mc S}=\{\bm{x}\in\mathbb{R}^n \mid h(\bm{x})\geq 0\}$. Then, CBF is defined as follows.
\begin{definition} \cite[Def. 5]{Ames2017} \label{def:CBF}
Given the control affine system 
\begin{align}
    \dot{\bm{x}} = f(\bm{x}) + g(\bm{x})\bm{u},
\end{align}
where $f$ and $g$ are locally Lipschitz, $\bm{x} \in \R^n$ and $\bm{u} \in \R^m$, together with the set $\mc S$. Then, the function $h$ is a control barrier function (CBF) defined on a set $\bar{\mathcal{S}}$ with $\mathcal{S} \subseteq \bar{\mathcal{S}} \subset \mathbb{R}^n$, if there exists an extended class $\mathcal{K}$ function $\alpha$, such that
\begin{align}
    \sup_{\bm{u}}\left[L_fh(\bm{x})+L_gh(\bm{x})\bm{u}+\alpha(h(\bm{x}))\right]\geq 0,~~~\forall \bm{x}\in\bar{\mathcal{S}}  \label{CBFdefine}
\end{align}
where $L_f h(\bm{x})$ and $L_g h(\bm{x})$ are the Lie derivatives of $h$ along $f(\bm{x})$ and $g(\bm{x})$, respectively.
\end{definition}

The forward invariance of the set $\mc S$ is ensured through the following corollary.
\begin{corollary}\cite[Cor. 2]{Ames2017} 
Given a set $\mc S$, if $h$ is a CBF on $\bar{\mc S}$, then any Lipschitz continuous controller $u(\bm{x}): \bar{\mc S} \to U$ such that 
\begin{align} \label{eq:CBF_cond}
L_f h(\bm{x}) + L_g h(\bm{x})\bm{u}(\bm{x}) + \alpha( h(\bm{x}) ) \geq 0,
\end{align}
will render the set $\mc S$ forward invariant.
\end{corollary}

The condition \eqref{eq:CBF_cond} guaranteeing forward invariance of the set $\mc S$ can be synthesized to the control law through the optimization-based controller leveraging Quadratic Programming (QP). Let us denote the nominal input as $\bm{u}_{\mathrm{nom}}$ and wish to modify it minimally invasive way so as to satisfy the condition \eqref{eq:CBF_cond}. This goal can be achieved by employing the input $\bm{u}^*$ derived from the following QP
\begin{subequations} \label{eq:QP_default}
\begin{align}
    \bm{u}^* =& \argmin_{\bm{u}}~\|\bm{u}-\bm{u}_{\mathrm{nom}}(x)\|^2,\\
    &~\mbox{s.t.}~L_fh(\bm{x})+L_gh(\bm{x})\bm{u}+\alpha(h(\bm{x}))\geq 0.
\end{align}
\end{subequations}

\section{Collision Avoidance for Elliptical Agents}

In this section, we formulate a novel CBF that ensures the forward invariance of the set $\mc S_{ij}$ in \eqref{eq:safeset}, namely preventing agent $i$ from colliding with agent $j$. 
As mentioned previously and from \cite{Do2013, Dimarogonas2019}, it is difficult to derive the analytical solution of the distance between two ellipses, namely $w_{ij}^*$, in a form simple enough to be employed as a CBF. Furthermore, numerical solutions of $w_{ij}^*$ cannot be employed as a CBF. To mitigate the difficulties, we design a novel CBF that incorporates a signed distance from a supporting line of agent $i$ to agent $j$, depicted as $h_{ij}$ in Fig.~\ref{fig:ellipse_setting}. Because $h_{ij}$ could take a shorter length than $w_{ij}^*$ with naive choices of the supporting line, we propose the procedure that drives $h_{ij}$ to $w_{ij}^*$ based on the gradient of an optimization problem.

\subsection{Separation Conditions for Two Elliptical Agents} \label{sec:distance}


\begin{figure}
    \centering
    \includegraphics[width=40mm]{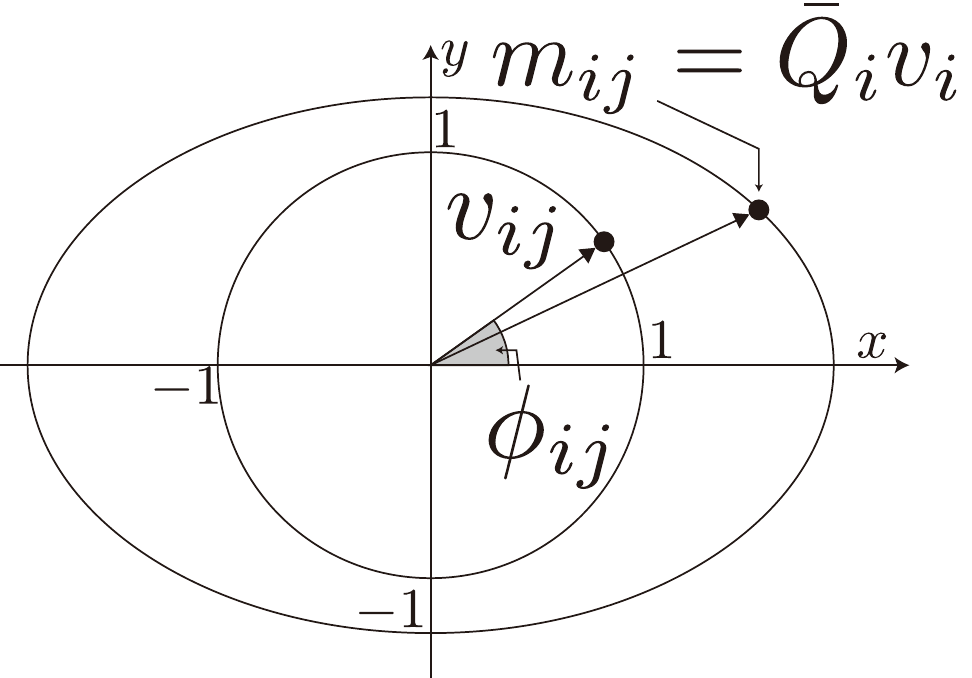}
    \caption{The parameter $\phi_{ij}$, which specifies a point on the boundary of the ellipse. The point on the unit circle is specified by $\bm{v}_{ij}$ in \eqref{eq:v_ij}, which angle from the $x$-axis is $\phi_{ij}$. Then, a point on the ellipse is specified by transforming $\bm{v}_{ij}$ with the positive definite matrix $\bar{Q}_i$.}
    \label{fig:phi_explanation}
\end{figure}

We first introduce a supporting line $l_{ij}$ of agent $i$, which contacts $\mc E_i$ at the point $\bm{m}_{ij}$, as depicted in Fig.~\ref{fig:ellipse_setting}. The point $\bm{m}_{ij}$ can be defined as
\begin{align}
    \bm{m}_{ij}(\bm{x}_i,\phi_{ij})&=\bar{Q}_i\bm{v}_{ij}+\bm{p}_i  \label{eq:m_ij}\\
    \bm{v}_{ij}(\phi_{ij})&=\left[\cos(\phi_{ij}),\sin(\phi_{ij})\right]^T \label{eq:v_ij}
\end{align}
with a positive definite matrix $\bar{Q}_i = R_iQ_iR_i^T$ and a parameter $\phi_{ij}\in(-\pi,\pi]$.
Here, the parameter $\phi_{ij}$ is introduced to specify a point on the boundary of the ellipse $\mc E_i$, where the graphical interpretation is shown in Fig.~\ref{fig:phi_explanation}. 
Then, the supporting line $l_{ij}$ is described as 
\begin{align}
    l_{ij} \!=\! \left\{X\in\mathbb{R}^2~|~\bm{v}_{ij}^T \bar{Q}_i^{-1} X-\left(1+\bm{v}_{ij}^T \bar{Q}_i^{-1} \bm{p}_i\right) \!=\! 0\right\}, \label{hyperplane}
\end{align}
which is determined by $\bm{x}_i$ and $\phi_{ij}$.

\begin{figure}
    \centering
    \begin{minipage}[b]{0.32\hsize}
        \centering
        \includegraphics[width=26mm]{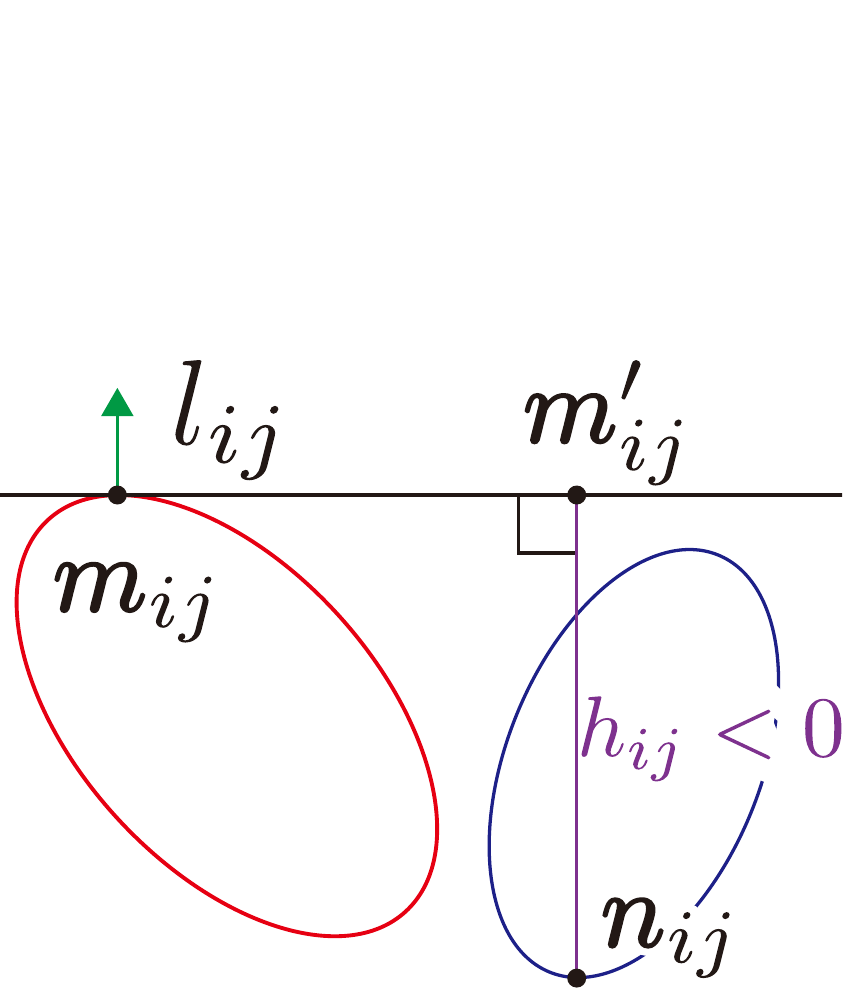}
        \subcaption{$h_{ij}(\phi_{ij})<0$}\label{subfig:h_explanation_1}
    \end{minipage}
    \begin{minipage}[b]{0.32\hsize}
        \centering
        \includegraphics[width=26mm]{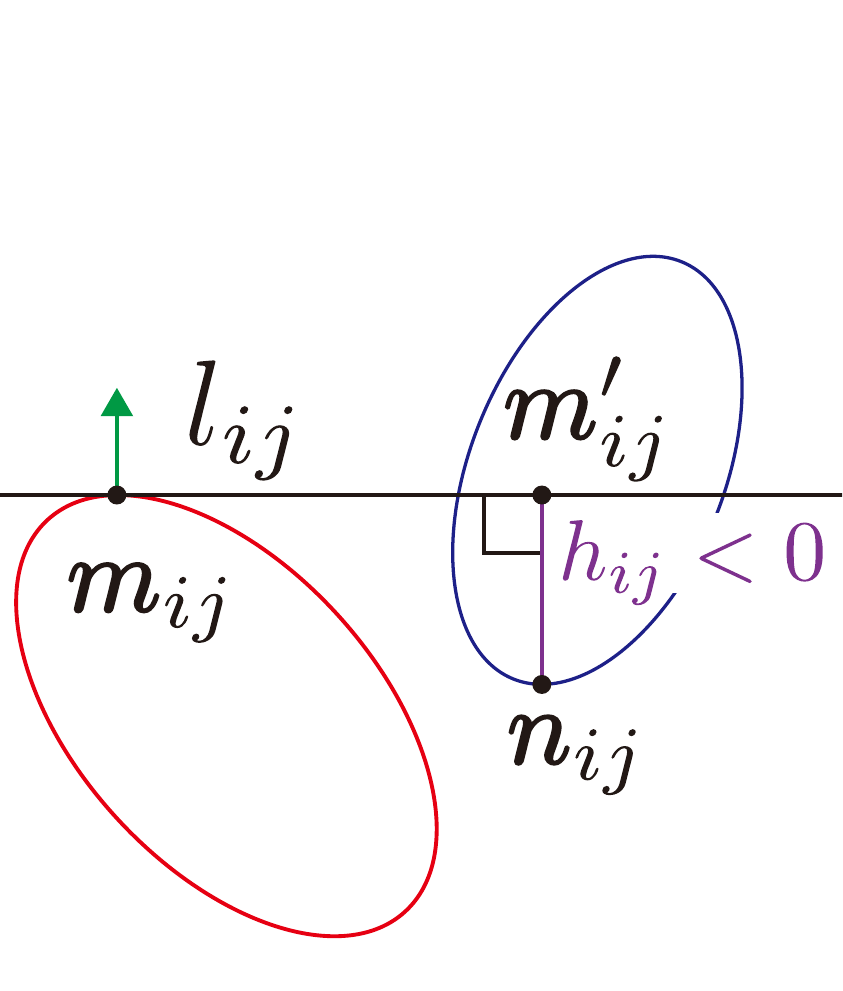}
        \subcaption{$h_{ij}(\phi_{ij})<0$}\label{subfig:h_explanation_2}
    \end{minipage}
    \begin{minipage}[b]{0.32\hsize}
        \centering
        \includegraphics[width=26mm]{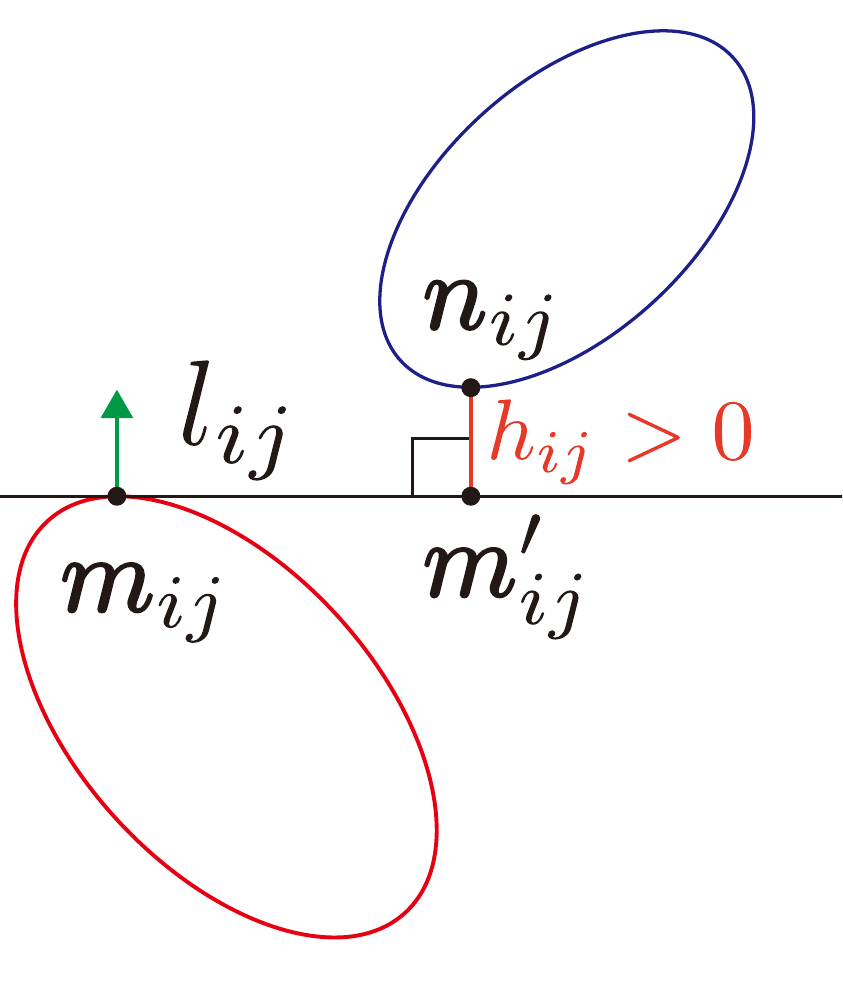}
        \subcaption{$h_{ij}(\phi_{ij})>0$}\label{subfig:h_explanation_3}
    \end{minipage}
    \caption{
    The separation condition evaluated by $h_{ij}$, the minimum signed distance from the line $l_{ij}$. 
    The signed distance $h_{ij}$ takes zero on $l_{ij}$ and takes the larger value as a point to be evaluated moves to the upper direction specified by the green normal vector. 
    (a) and (b): Since the proposed signed distance provides a negative value to a point in the same half-plane with the ellipse $\mc E_i$, $\bm{n}_{ij}\in \mc E_j$ equipped with the minimum signed distance $h_{ij}$ is the one furthest from $l_{ij}$. (c): When the supporting line separates two ellipses, $h_{ij}$ returns the distance between $l_{ij}$ and $\mc E_j$.
    }
    \label{fig:h_explanation}
\end{figure}


Let us derive the separation condition evaluated with the signed distance from the supporting line $l_{ij}$, which provides a positive value to a point in the different half-plane with $\mc E_i$, and a negative value otherwise as shown in Fig.~\ref{fig:h_explanation}.
The point $\bm{n}_{ij} \in \mc E_j$ that minimizes the signed distance from the supporting line $l_{ij}$ is determined by
\begin{align}
    \bm{n}_{ij}(\bm{x}_i,\bm{x}_j,\phi_{ij}) = -\frac{1}{\left\|\bar{Q}_j \bar{Q}_i^{-1}\bm{v}_{ij}\right\|}\bar{Q}_j^2\bar{Q}_i^{-1}\bm{v}_{ij} + \bm{p}_j.
\end{align}
Then, the minimum signed distance from $l_{ij}$ is calculated by
\begin{align}
    h_{ij}(\bm{x}_i,\bm{x}_j,\phi_{ij})\!=\!\frac{ -\left\|\bar{Q}_j \bar{Q}_i^{-1} \bm{v}_{ij}\right\| \!+\! (\bm{p}_j-\bm{p}_i)^T\bar{Q}_i^{-1} \bm{v}_{ij} \!-\! 1}{\left\|\bar{Q}_i^{-1}\bm{v}_{ij}\right\|}, \label{ThisCBF}
\end{align}
which satisfies $h_{ij}(\bm{x}_i,\bm{x}_j,\phi_{ij}) > 0$ if and only if $\mc E_i$ and $\bm{n}_{ij}$ are in the different half-plane separated by the line $l_{ij}$ as shown in Fig.~\ref{fig:h_explanation}.
In other words, if there exists $\phi_{ij}$ that fulfills $h_{ij}(\bm{x}_i,\bm{x}_j,\phi_{ij}) > 0$, then two ellipses are separated by the line $l_{ij}$.
Note that $\bm{n}_{ij} \in \mc E_j$ is not the closest point to the supporting line $l_{ij}$ as depicted in Fig.~\ref{fig:h_explanation}(a) and (b).
In the remaining of this subsection, we analyze the property of $h_{ij}(\bm{x}_i,\bm{x}_j,\phi_{ij})$ assuming two agents $i$ and $j$ are fixed, hence denote $h_{ij}(\bm{x}_i,\bm{x}_j,\phi_{ij})$ as $h_{ij}(\phi_{ij})$.

The fact that $h_{ij}(\phi_{ij}) > 0$ encodes the separation condition between two agents motivates us to employ $h_{ij}(\phi_{ij})$ as a CBF for ensuring collision avoidance between elliptical agents.
However, the naive selection of the parameter $\phi_{ij}$, namely the line $l_{ij}$, could yield a shorter distance than the actual distance $w_{ij}^*$ as depicted in Fig.~\ref{fig:ellipse_setting}(a). 
Because of this difference from $w_{ij}^*$, $h_{ij}(\phi_{ij})$ could overestimate the risk of collisions and hence cause too conservative evasive motion.
To mitigate this gap, we propose the following optimization problem that intends to rotate the supporting line $l_{ij}$ on the boundary of $\mc E_i$ so that $h_{ij}(\phi_{ij})$ is maximized as
\begin{subequations}\label{eq:dual}
\begin{align}
    \max_{\phi_{ij}}~& h_{ij}(\phi_{ij}),\\
    \mbox{s.t.}~& \bm{v}_{ij}=\left[\cos(\phi_{ij}),\sin(\phi_{ij})\right]^T,
\end{align}
\end{subequations}
where its graphical interpretation is depicted in Fig.~\ref{fig:ellipse_setting}(b).


While, in the next subsection, we introduce the input to $\phi_{ij}$ for maximizing $h_{ij}(\phi_{ij})$, we first establish the connection between the optimization problem \eqref{eq:dual} and the actual distance between two agents, 
namely $w_{ij}^*$.
For this goal, we introduce the following optimization problem, which optimal solution $\|\bm{w}^*\|$ is equal to $w_{ij}^*$.
\begin{subequations}\label{eq:primal}
\begin{align}
    \min_{\bm{\xi},\bm{\eta},\bm{w}}~&\|\bm{w}\|,\\
    \mbox{s.t.}~&f_i(\bm{\xi})\leq 0,~f_j(\bm{\eta})\leq 0,\\
    &\bm{\eta}-\bm{\xi}=\bm{w},
\end{align}
\end{subequations}
where $f_k(\bm{\xi}):=(\bm{\xi}-\bm{p}_k)^T\bar{Q}_k^{-2}(\bm{\xi}-\bm{p}_k)-1\leq 0$ signifies the condition $\bm{\xi} \in \mc E_k$.
Then, the following theorem formalizes the relationship between $w_{ij}^*$ and $h_{ij}(\phi_{ij})$.
\begin{theorem} \label{th:Prim_dual}
Suppose that two ellipses $\mc E_i$ and $\mc E_j$ have no overlap, namely $\mc E_i \cap \mc E_j = \emptyset$ holds.
Then, the optimization problem \eqref{eq:dual} is the dual of the problem \eqref{eq:primal}.
Furthermore, the strong duality holds between the optimization problems \eqref{eq:primal} and \eqref{eq:dual}, namely the following condition holds
\begin{align} \label{eq:no_gap}
    w_{ij}^* = h_{ij}^* \geq h_{ij}(\phi_{ij}).
\end{align}
\end{theorem}

\begin{proof}
The dual function of the problem \eqref{eq:primal} is
\begin{align}
\begin{aligned}
g(\lambda_i,\lambda_j,\bm{z}) = \inf_{\bm{\xi},\bm{\eta},\bm{w}}\left( \|\bm{w}\| \right.&\left.+\lambda_if_i(\bm{\xi}) + \lambda_jf_j(\bm{\eta}) \right.\\
&\left. +\bm{z}^T(\bm{\eta}-\bm{\xi}-\bm{w})  \right)
\end{aligned}\\
=\begin{cases}
        \begin{aligned}
        \inf_{\bm{\xi}} &\left(\lambda_i f_i(\bm{\xi}) - \bm{z}^T\bm{\xi}\right)\\&+ \inf_{\bm{\eta}} \left(\lambda_j f_j(\bm{\eta}) + \bm{z}^T\bm{\eta}\right)
        \end{aligned} & \begin{aligned}
        \|\bm{z}\|&\leq 1,\\\lambda_i,\lambda_j &\geq 0
        \end{aligned}\\
    -\infty & \mbox{otherwise}
    \end{cases}, \label{eq:Lagrangian}
\end{align}
where $\lambda_i$, $\lambda_j$ and $\bm{z}$ are Lagrange multipliers. 
To simplify the first term $\inf_{\bm{\xi}} \left(\lambda_i f_i(\bm{\xi}) - \bm{z}^T\bm{\xi}\right)$ in \eqref{eq:Lagrangian}, we introduce $\bar{\bm{\xi}}=\bar{Q}_i^{-1}\bm{\xi}$, $\bar{\bm{p}}_i=\bar{Q}_i^{-1}\bm{p}_i$, and $\bar{\bm{z}}=\bar{Q}_i \bm{z}$.
Then, the first term is transformed as
\begin{align}
    &\inf_{\bm{\xi}} \left(\lambda_i f_i({\bm{\xi}}) - \bm{z}^T{\bm{\xi}}\right)\nonumber\\
    =&\inf_{\bm{\xi}}\left( \lambda_i({\bm{\xi}}-\bm{p}_i)^T\bar{Q}_i^{-2}({\bm{\xi}}-\bm{p}_i)-\lambda_i - \bm{z}^T{\bm{\xi}}\right)\nonumber\\
    =&-\bar{\bm{z}}^T\bar{\bm{p}}_i-\frac{\|\bar{\bm{z}}\|^2+4\lambda_i^2}{4\lambda_i}. \label{eq:dual_first}
\end{align}
Following the similar path to \eqref{eq:dual_first}, the second term in \eqref{eq:Lagrangian} is expressed as
\begin{align}
    \inf_{\bm{\eta}} \left(\lambda_j f_j(\bm{\eta}) + \bm{z}^T\bm{\eta}\right)=\hat{\bm{z}}^T\hat{\bm{p}}_j-\frac{\|\hat{\bm{z}}\|^2+4\lambda_j^2}{4\lambda_j},
\end{align}
where $\hat{\bm{p}}_j=\bar{Q}_j^{-1}\bm{p}_j$ and $\hat{\bm{z}}=\bar{Q}_j \bm{z}$.
Therefore, the dual problem can be expressed as follows:
\begin{subequations}\label{eq:dual_1}
\begin{align}
    \max_{z,\lambda_i,\lambda_j}~&-\bar{\bm{z}}^T\bar{\bm{p}}_i -\frac{\|\bar{\bm{z}}\|^2+4\lambda_i^2}{4\lambda_i}
    +\hat{\bm{z}}^T\hat{\bm{p}}_j-\frac{\|\hat{\bm{z}}\|^2+4\lambda_j^2}{4\lambda_j},\label{eq:dual_1_a}\\
    \mbox{s.t.}~&\bar{\bm{z}}=Q_i\bm{z},~
    \hat{\bm{z}}=Q_j \bm{z},~
    \|\bm{z}\|\leq 1,~
    \lambda_i,\lambda_j\geq 0.
\end{align}
\end{subequations}

Focusing on the second and fourth terms in \eqref{eq:dual_1_a}, we define a function $M(a,x)$ as 
\begin{align}
    M(a,x)=-\frac{a+4x^2}{4x}~~~(a,x\geq 0),
\end{align}
and consider $\max_{a,x}M(a,x)$. 

In the case of $a>0$, the gradient of $M$ is 
\begin{align}
    \frac{\partial M}{\partial x} = \frac{(\sqrt{a}-2x)(\sqrt{a}+2x)}{4x^2},~\frac{\partial M}{\partial a} = - \frac{1}{4x}.
\end{align}
Thus, the function $M(a,x)$ has no extremum for all $a>0$, and for $x\geq0$ it has the maximum value at $x^*(a)={\sqrt{a}}/{2}$. 
As a result, the following equation holds. 
\begin{align}
    \max_{a,x}M(a,x)
    =\max_{a}-\sqrt{a}~~~(a>0).
\label{eq:a_neq_0}
\end{align}

In the case of $a=0$, the maximum value of $M(a,x)$ is as follows:
\begin{align}
    \max_{a,x} M(a,x) = \max_x-x = 0.
    \label{eq:a_eq_0}
\end{align}
%
Since \eqref{eq:a_neq_0} is equivalent to \eqref{eq:a_eq_0} if we substitute $a=0$ into \eqref{eq:a_neq_0}, we can summarize them as 
\begin{align}
    \max_{a,x} M(a,x)=\max_{a}-\sqrt{a}~~~(a\geq 0).
    \label{eq:maxM_a_x}
\end{align}
Considering that the second and fourth terms in \eqref{eq:dual_1_a} are equal to $M(\|\bar{\bm{z}}\|^2,\lambda_i)$ and $M(\|\hat{\bm{z}}\|^2,\lambda_j)$, respectively, we can simplify the problem \eqref{eq:dual_1} by utilizing \eqref{eq:maxM_a_x} as follows:
\begin{subequations}\label{eq:dual_2}
\begin{align}
    \max_{\bm{z}}~&-\bar{\bm{z}}^T\bar{\bm{p}}_i-\left\|\bar{\bm{z}}\right\| + \hat{\bm{z}}^T\hat{\bm{p}}_j-\left\|\hat{\bm{z}}\right\| ,\\
      \mbox{s.t.}~&\|\bm{z}\|\leq 1, \label{eq:dual_2_const2}\\
      &\bar{\bm{z}}=\bar{Q}_i \bm{z},~\hat{\bm{z}}=\bar{Q}_j \bm{z}.
\end{align}
\end{subequations}

Let us parameterize $\bm{z}$ as 
\begin{align}
    \bm{z}&=\frac{\mu}{\left\|\bar{Q}_i^{-1}\bm{v}_{ij}\right\|}\bar{Q}_i^{-1}\bm{v}_{ij},\label{eq:z_changed}
\end{align}
with $0\leq\mu\leq 1$ so that the constraint \eqref{eq:dual_2_const2} is satisfied.
Then, by substituting \eqref{eq:z_changed}, $\bar{\bm{z}} = \bar{Q}_i \bm{z}$ and $\hat{\bm{z}} = \bar{Q}_j \bm{z}$, the dual problem \eqref{eq:dual_2} can be transformed  as follows
\begin{subequations}\label{eq:dual_3}
\begin{align}
    \max_{\mu,\phi_{ij}}~&\mu\underbrace{\frac{ -\left\|\bar{Q}_j \bar{Q}_i^{-1} \bm{v}_{ij}\right\| + (\bm{p}_j-\bm{p}_i)^T\bar{Q}_i^{-1} \bm{v}_{ij} - 1}{\left\|\bar{Q}_i^{-1}\bm{v}_{ij}\right\|}}_{h_{ij}(\phi_{ij})},\label{eq:dual_3_a}\\
    \mbox{s.t.}~& \bm{v}_{ij}=[\cos(\phi_{ij}),\sin(\phi_{ij})]^T,
    \\&0\leq \mu \leq 1,\label{eq:dual_3_c}
\end{align}
\end{subequations}
where the objective function \eqref{eq:dual_3_a} can be described as $\mu h_{ij}(\phi_{ij})$. 
Notice that if ellipses $\mathcal{E}_i$ and $\mathcal{E}_j$ are separated, there always exists $\phi_{ij}$ that satisfies $h_{ij}(\phi_{ij})>0$. Considering the constraint \eqref{eq:dual_3_c} and the fact $\mu$ is a variable independent of $\phi_{ij}$, we should set $\mu=1$ to maximize \eqref{eq:dual_3_a}. 
Therefore, by substituting $\mu=1$ to the problem \eqref{eq:dual_3}, we obtain the optimization problem \eqref{eq:dual}.
This relation leads to the conclusion that the problem \eqref{eq:dual} is the dual of the optimization problem \eqref{eq:primal} if $\mc E_i \cap \mc E_j = \emptyset$.
Since the optimization problem \eqref{eq:primal} satisfies the Slater's Condition \cite[Sec. 5.2.3]{convex_optimization}, the solution of \eqref{eq:primal} is equivalent to the solution of \eqref{eq:dual}, which completes the proof. 
\end{proof}

Theorem~\ref{th:Prim_dual} implies that the proposed update law of the supporting line, described as the optimization problem \eqref{eq:dual} and Fig.~\ref{fig:ellipse_setting}, renders $h_{ij}(\phi_{ij})$ the actual distance $w_{ij}^*$ between two ellipses. Furthermore, because of \eqref{eq:no_gap}, the condition $h_{ij}(\phi_{ij})>0$ provides a sufficient condition for preventing collisions even if $\phi_{ij}$ does not converge to the maximizer of \eqref{eq:dual}. 
In the following subsection, we develop $h_{ij}$ as a CBF together with presenting the update procedure of $\phi_{ij}$. 


\subsection{CBFs Incorporating Rotating Supporting Lines} \label{sec:CBF}
In this subsection, we first design the collision avoidance methods between two elliptical agents $i$ and $j$, then extend the results to the case of $n$ agents. 
Hereafter, we regard $h_{ij}$ as a function of $\bm{x}_i, \bm{x}_j$, and $\phi_{ij}$ as first defined in \eqref{ThisCBF}, though we have dropped the dependency of $h_{ij}$ for notational convenience.

As mentioned in the previous subsection, we require a supporting line between agents $i$ and $j$ to evaluate the separation condition $h_{ij}$. Without loss of generality, we introduce a supporting line for the agent with a lower ID number. Then, as the supporting line is to be updated for maximizing $h_{ij}$, the model of the agent \eqref{eq:dynamics} now has to describe the dynamics of $\phi_{ij}$ as well. Therefore, we introduce the augmented state $\bm{x}_{ij}=\left[\bm{x}_{i}^T,~\bm{x}_{j}^T,~\phi_{ij}\right]^T$ to achieve collision avoidance between two agents $i$ and $j$ ($i<j$), where the dynamics of the augmented state is 
\begin{align}
    \dot{\bm{x}}_{ij} = \bm{u}_{ij}, \label{eq:aug_syst_ij}
\end{align}
where $\bm{u}_{ij} = \left[\bm{u}_i^T, \bm{u}_j^T, u_{\phi_{ij}} \right]^T$.

As the nominal input for $\phi_{ij}$ intending to maximize $h_{ij}$, we employ the following gradient-based input 
\begin{align}
    u_{\mathrm{nom},\phi_{ij}}=\gamma \frac{\partial h_{ij}}{\partial \phi_{ij}},~~~\gamma >0. \label{dphi}
\end{align}
The input \eqref{dphi} drives $h_{ij}$ to a local maximum point, and hence preventing too conservative evasive motions caused by the difference between $w_{ij}^*$ and $h_{ij}$. 
Note that the maximizer $\phi_{ij}^*$ of $h_{ij}$ changes as the agent $i$ and $j$ traverse.
However, as $\phi_{ij}$ is a virtual variable not dependent on the physical dynamics of the agents, we can make $\phi_{ij}$ converge fast enough with large $\gamma$ to follow the agent movement, as we will demonstrate in the simulations in Section~\ref{subsec:two_simu}.



Let us denote the nominal control input for the agent~$i$ as $\bm{u}_{{\rm nom}, i}$, which is designed to achieve its own objective, e.g., reaching the goal position. Combining this nominal control input $\bm{u}_{{\rm nom}, i}$ with \eqref{dphi}, the nominal input for the augmented state $\bm{x}_{ij}$ is expressed as
$\bm{u}_{\mathrm{nom},ij}=\left[\bm{u}_{\mathrm{nom},i}^T,~\bm{u}_{\mathrm{nom},j}^T,~u_{\mathrm{nom},\phi_{ij}}\right]^T$.

The goal of the collision avoidance strategy is to allow agents to execute their tasks encoded by $\bm{u}_{{\rm nom}, i}$ while ensuring collision never occurs between agents. 
To achieve this objective, we propose the collision avoidance methods that render the following set $\hat{\mc S}_{ij}$ forward invariant.
\begin{align}
    \hat{\mc S}_{ij} = \left\{\bm{x}_{ij} \in \mathbb{R}^6\times (-\pi, \pi]  \mid h_{ij}(\bm{x}_{ij})\geq 0\right\}
\end{align}
Since $w_{ij}^* \geq h_{ij}$ holds from Theorem~\ref{th:Prim_dual}, the set $\mc S_{ij}$ in \eqref{eq:safeset} is also ensured to be forward invariant if we guarantee the forward invariance of $\hat{\mc S}_{ij}$.
Similar to \eqref{eq:QP_default}, this strategy can be achieved by $\bm{u}_{ij}^*$ derived from the following QP integrating $h_{ij}$ as a CBF:
\begin{subequations} \label{eq:QP_two}
\begin{align}
    &\bm{u}_{ij}^* = \argmin_{\bm{u}_{ij}}~\left\|\bm{u}_{ij}-\bm{u}_{\mathrm{nom},ij}\right\|^2,\\
    &\mbox{s.t.}~{\frac{\partial h_{ij}}{\partial \bm{x}_{i}}}^T\bm{u}_i +{\frac{\partial h_{ij}}{\partial \bm{x}_{j}}}^T\bm{u}_j +\frac{\partial h_{ij}}{\partial \phi_{ij}}u_{\phi_{ij}} + \alpha(h)\geq 0.\label{trueCBFcond}
\end{align}
\end{subequations}
Note that the constraint \eqref{trueCBFcond} is derived by calculating the Lie derivative of $h_{ij}$ along $f_{ij}=\bm{0}_{7\times 1}$ and $g_{ij}=I_7$ in \eqref{eq:aug_syst_ij}.

As the proposed CBF $h_{ij}$ prevents collision between elliptical agents $i$ and $j$, we need to confirm whether there always exists the control input that satisfies \eqref{CBFdefine} in the Definition~\ref{def:CBF} and it provides us the forward invariance property. 
\begin{theorem} \label{th:valid_CBF}
The function $h_{ij}$ in \eqref{ThisCBF} is a valid CBF when $\bm{u}_i,\bm{u}_j \in \R^3$. Namely, for any $\bm{x}_{i},\bm{x}_{j}\in\mathbb{R}^3$, the following conditions hold. 
\begin{align}
    \frac{\partial h_{ij}}{\partial \bm{x}_{i}}&=\left[{\frac{\partial h_{ij}}{\partial \bm{p}_i}}^T,\frac{\partial h_{ij}}{\partial \theta_i}\right]^T\neq \bm{0}_{3\times 1}\\
    \frac{\partial h_{ij}}{\partial \bm{x}_{j}}&=\left[{\frac{\partial h_{ij}}{\partial \bm{p}_j}}^T,\frac{\partial h_{ij}}{\partial \theta_j}\right]^T\neq \bm{0}_{3\times 1}
\end{align}
\end{theorem}
\begin{proof}
${\partial h_{ij}}/{\partial \bm{p}_i}$ and ${\partial h_{ij}}/{\partial \bm{p}_j}$ are given by
\begin{align}
    \frac{\partial h_{ij}}{\partial \bm{p}_i} &= -\frac{1}{\left\|\bar{Q}_i^{-1}\bm{v}_{ij}\right\|}\bar{Q}_i^{-1}\bm{v}_{ij},\\
    \frac{\partial h_{ij}}{\partial \bm{p}_j} &= \frac{1}{\left\|\bar{Q}_i^{-1}\bm{v}_{ij}\right\|}\bar{Q}_i^{-1}\bm{v}_{ij}.
\end{align}
Considering the fact $\left\|\bm{v}_{ij}\right\|=1$ and $\bar{Q}_i\succ 0$, 
\begin{align}
    \left\|\frac{\partial h_{ij}}{\partial \bm{p}_i}\right\|=\left\|\frac{\partial h_{ij}}{\partial \bm{p}_j}\right\|=1 \nonumber
\end{align}
holds.
Therefore, the conditions ${\partial h_{ij}}/{\partial \bm{x}_{i}}\neq \bm{0}_{3\times 1}$ and ${\partial h_{ij}}/{\partial \bm{x}_{j}}\neq \bm{0}_{3\times 1}$ are always satisfied. 
\end{proof}

Theorem~\ref{th:valid_CBF} signifies there always exists the control input $\bm{u}_i,\bm{u}_j \in \R^3$ that renders the set $\hat{\mc S}_{ij}$ forward invariant.
Therefore, the collision is prevented if the state $\bm{x}_{ij}$ is in the safe set $\hat{\mc S}_{ij}$ at the initial time. 
Note that if two ellipses do not collide with each other at the initial time, we can easily specify the initial $\phi_{ij}$ that satisfies $h_{ij}>0$ by any heuristic approach.




We then extend the collision avoidance method \eqref{eq:QP_two} to the scenario with $n$ agents. Similar to the case of the two agents, we need to introduce a supporting line for each pair of agents, resulting in $_{n}C_{2}$ numbers of CBFs for the whole system. Since the agent with a smaller ID in a pair owns a supporting line, a vector augmenting $\phi_{ij}$ of agent $i$ is described as $\bm{\phi}_{i}=[\phi_{i\hspace{0.4mm}i+1},\phi_{i\hspace{0.4mm}i+2},...,\phi_{i\hspace{0.4mm}n}]^T$ with $\bm{\phi}_n=\emptyset$.
In addition, we introduce a vector combining $h_{ij}$ as $\bm{h}_{i}=[h_{i\hspace{0.4mm}i+1},h_{i\hspace{0.4mm}i+2},...,h_{i\hspace{0.4mm}n}]^T$.
Then, the ensemble state and CBFs of the system are expressed as $\bm{x}=\left[\bm{x}_1^T,\bm{x}_2^T,...,\bm{x}_n^T,\bm{\phi}_1^T,\bm{\phi}_2^T,...,\bm{\phi}_{n-1}^T\right]^T$ and $\bm{h}=\left[\bm{h}_{1}^T,\bm{h}_{2}^T,..,\bm{h}_{n-1}^T\right]^T$, respectively.
With the introduced ensemble vectors, the collision avoidance method for $n$ agents can be achieved by $\bm{u}^*$ derived from
\begin{subequations}
\begin{align}
    \bm{u}^* =& \argmin_{\bm{u}}\|\bm{u}-\bm{u}_{\mathrm{nom}}\|^2,\\
    &~\mbox{s.t. }L_f\bm{h}+L_g\bm{h}\bm{u}+\alpha(\bm{h})\geq 0,
\end{align}
\end{subequations}
where $f=\bm{0}_{(3n + _{n}C_{2}) \times 1}$ and $g_{ij}=I_{(3n + _{n}C_{2})}$ since we assume a single integrator model.

\section{Simulation Results} \label{sec:sim}

The proposed algorithm is implemented in simulations to verify that it guarantees collision avoidance between elliptical agents.

\subsection{Simulation with Two Agents}\label{subsec:two_simu}
We first demonstrate our proposed algorithm with the two elliptical agents, the sizes of which are characterized by $Q_1=\mathrm{diag}(0.4,0.2)$ for a red agent and $Q_2=\mathrm{diag}(0.6,0.2)$ for a blue agent, respectively. The initial condition of the simulation is depicted in Fig.~\ref{fig:simu1_state}(a), with the initial pose $x_{1}(0)=\left[0,1,{-\pi}/{4}\right]^T$,  $x_{2}(0)=\left[2,0.1,0\right]^T$.
Note that we randomly chose the initial $\phi_{12}(0)$ from the range that yields a supporting line $l_{12}$ separating two ellipses.
Two agents traverse the environment so that they intersect around the center of the field. We utilize $\alpha(h_{12})=10h_{12}$ for an extended class $\mathcal{K}$ function in CBF conditions and $u_{\mathrm{nom},\phi_{12}}=20\left({\partial h_{12}}/{\partial \phi_{12}}\right)$ for a nominal input to $\phi_{12}$.


The snapshots of the simulations are presented in Fig.~\ref{fig:simu1_state}, which illustrate a supporting line incorporated in agent~1's CBF as a black line. In addition, the distance between the supporting line and the blue ellipse is depicted as a green line. 
The snapshots demonstrate that the supporting line on agent $i$ is updated so that it separates two agents without causing any conservative transition in their evasive trajectories. The value of $h_{12}$ is illustrated in Fig.~\ref{fig:simu1_h_his} together with the optimal solution of \eqref{eq:primal}, namely the actual distance between two agents.
Although $w_{12}^*$ and $h_{12}$ differ at the initial time since we set $\phi_{12}$ randomly, the proposed gradient-based input \eqref{dphi} successfully drives $h_{12}$ to the actual distance $w_{12}^*$.
Furthermore, Fig.~\ref{fig:simu1_h_his} verifies that the input \eqref{dphi} rotates the supporting line so that $h_{12}$ follows the transition of the actual distance between two ellipses.
We can also confirm that the value of $h_{12}$ keeps in the positive value, hence, collision avoidance is achieved.

\begin{figure}[t!]
    \centering
    \subfloat[$t=0\,{\rm s}$]
    {\includegraphics[trim = 0.08cm 0.1cm 1.4cm 1.2cm, clip=true, width=0.4\columnwidth]{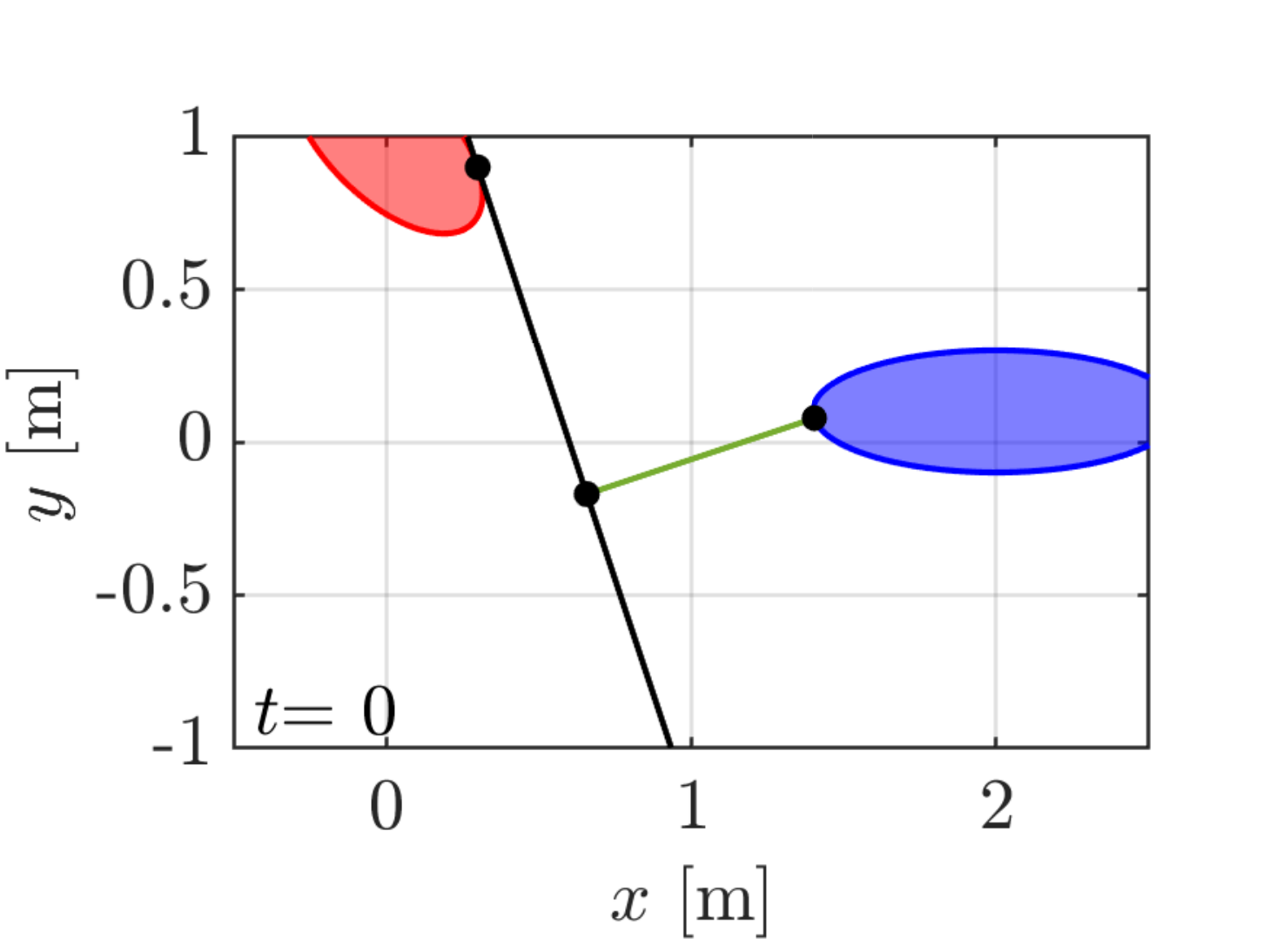}
    \label{subfig:simu1_state_1}} \quad 
    \subfloat[$t=0.8\,{\rm s}$]
    {\includegraphics[trim = 0.1cm 0cm 1cm 0.3cm, clip=true, width=0.4\columnwidth]{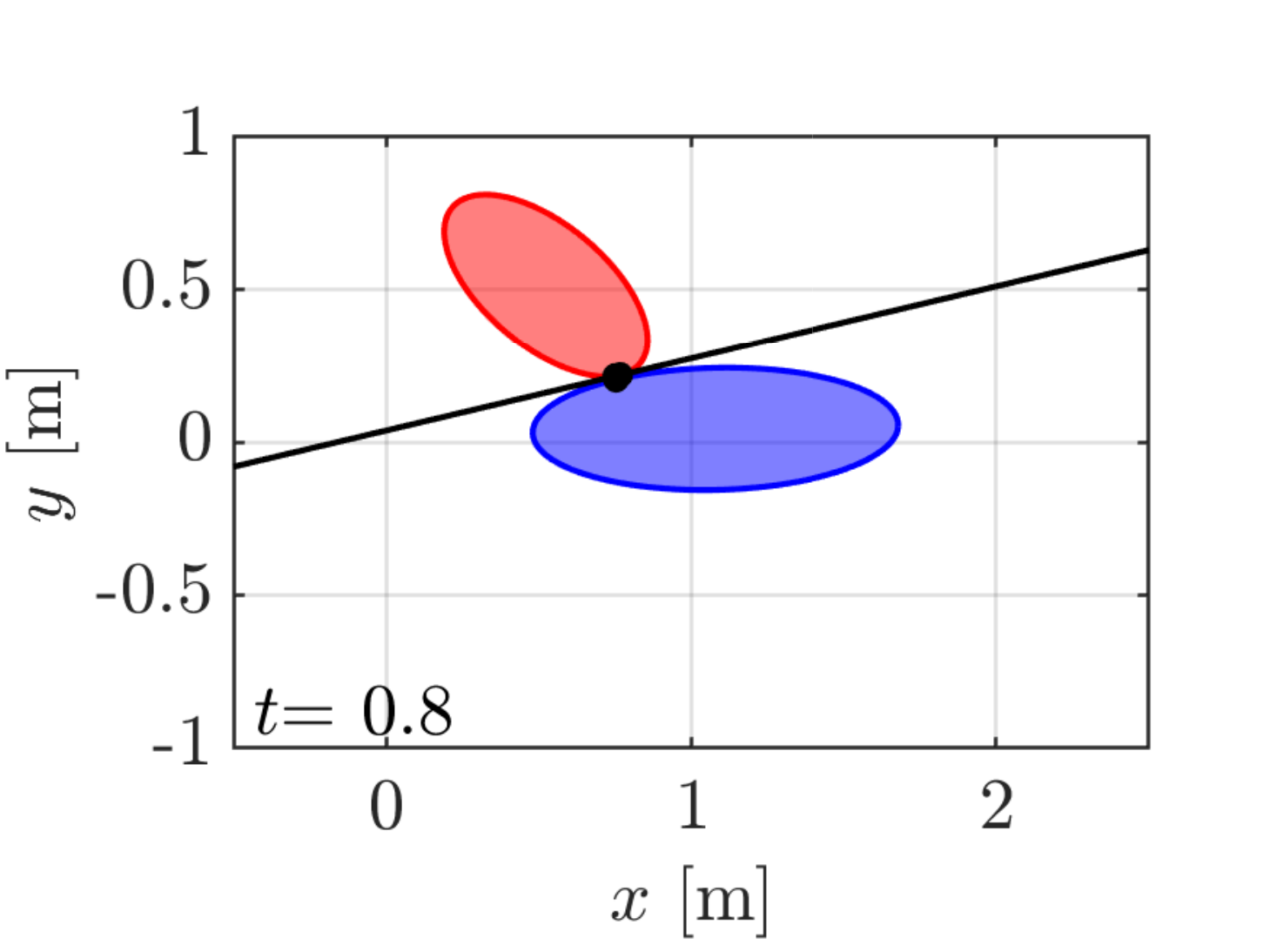}
    \label{subfig:simu1_state_2}} \\
    \subfloat[$t=2\,{\rm s}$]
    {\includegraphics[trim = 0.1cm 0cm 1cm 0.2cm, clip=true, width=0.4\columnwidth]{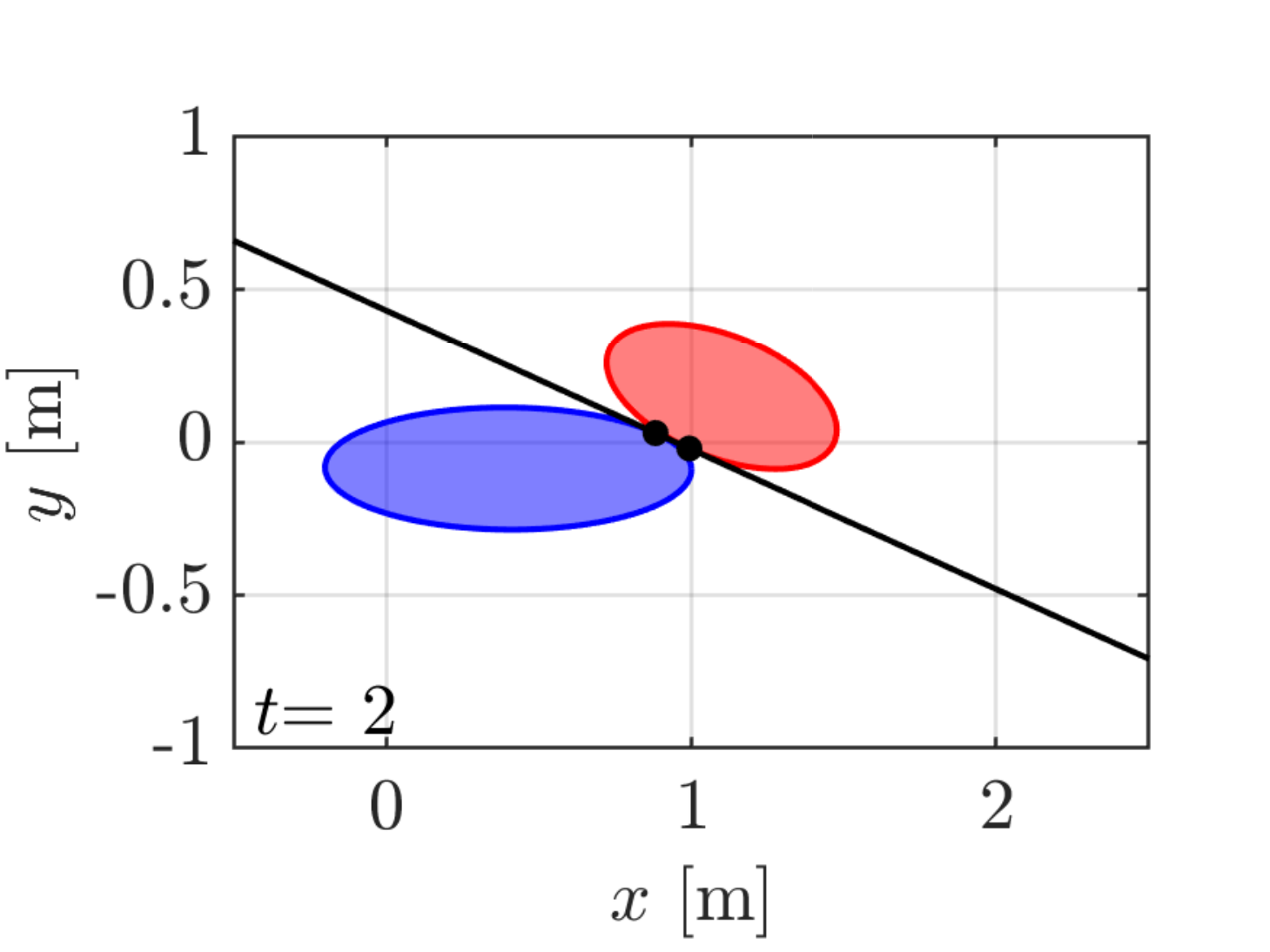}
    \label{subfig:simu1_state_3}} \quad 
    \subfloat[$t=4\,{\rm s}$]
    {\includegraphics[trim = 0.1cm 0cm 1cm 0.2cm, clip=true, width=0.4\columnwidth]{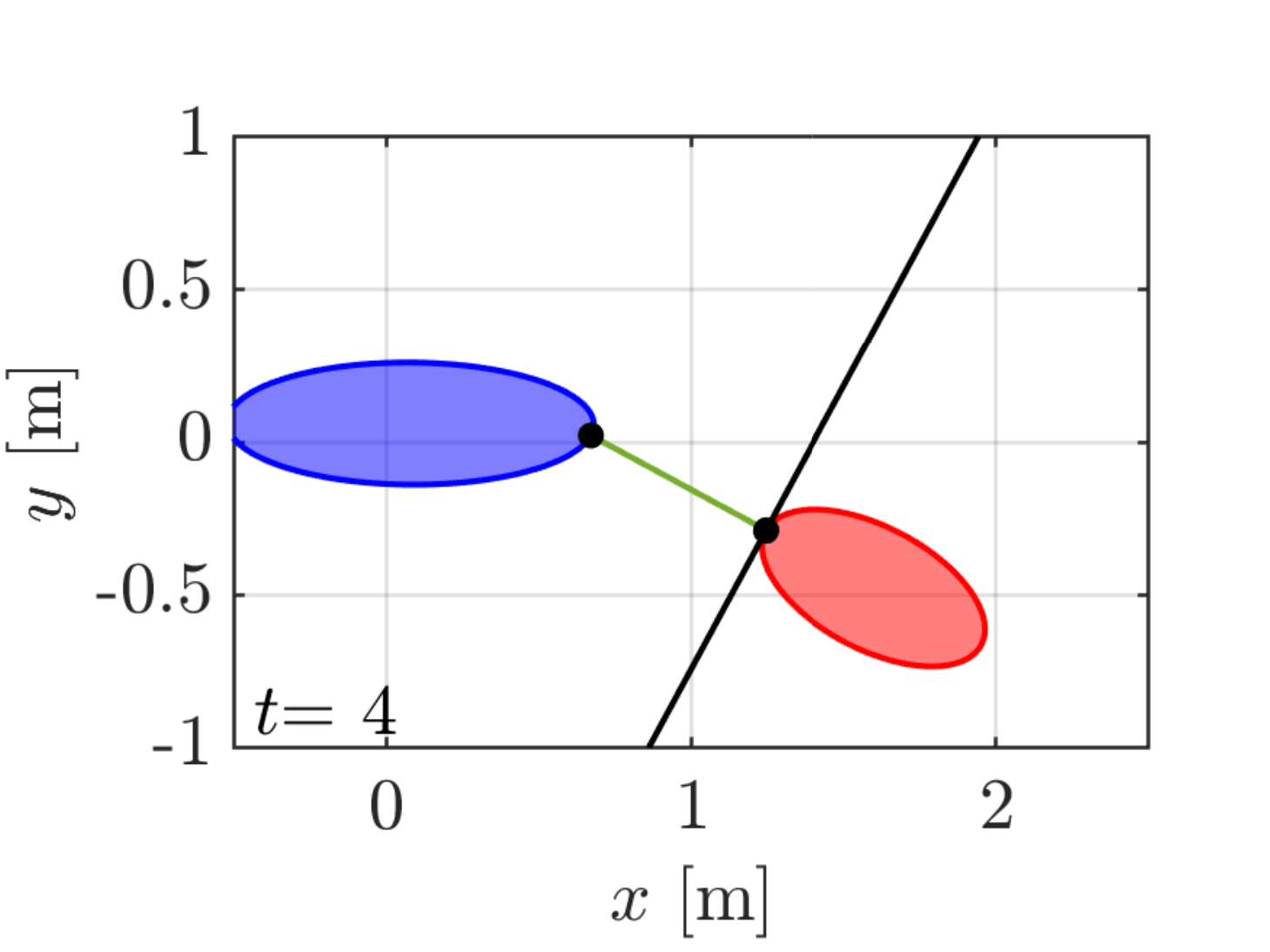}
    \label{subfig:simu1_state_4}} 
    \caption{Snapshots of the simulation, where two elliptical agent 1 and 2 are depicted in red and blue, respectively. The supporting line of agent~1 is rotated to maximize the distance, shown in the green line, between the line and agent 2.}
    \label{fig:simu1_state}
\end{figure}

\begin{figure}
    \centering
    \includegraphics[trim = 0cm 0.1cm 0cm 0.2cm, clip=true, width=45mm]{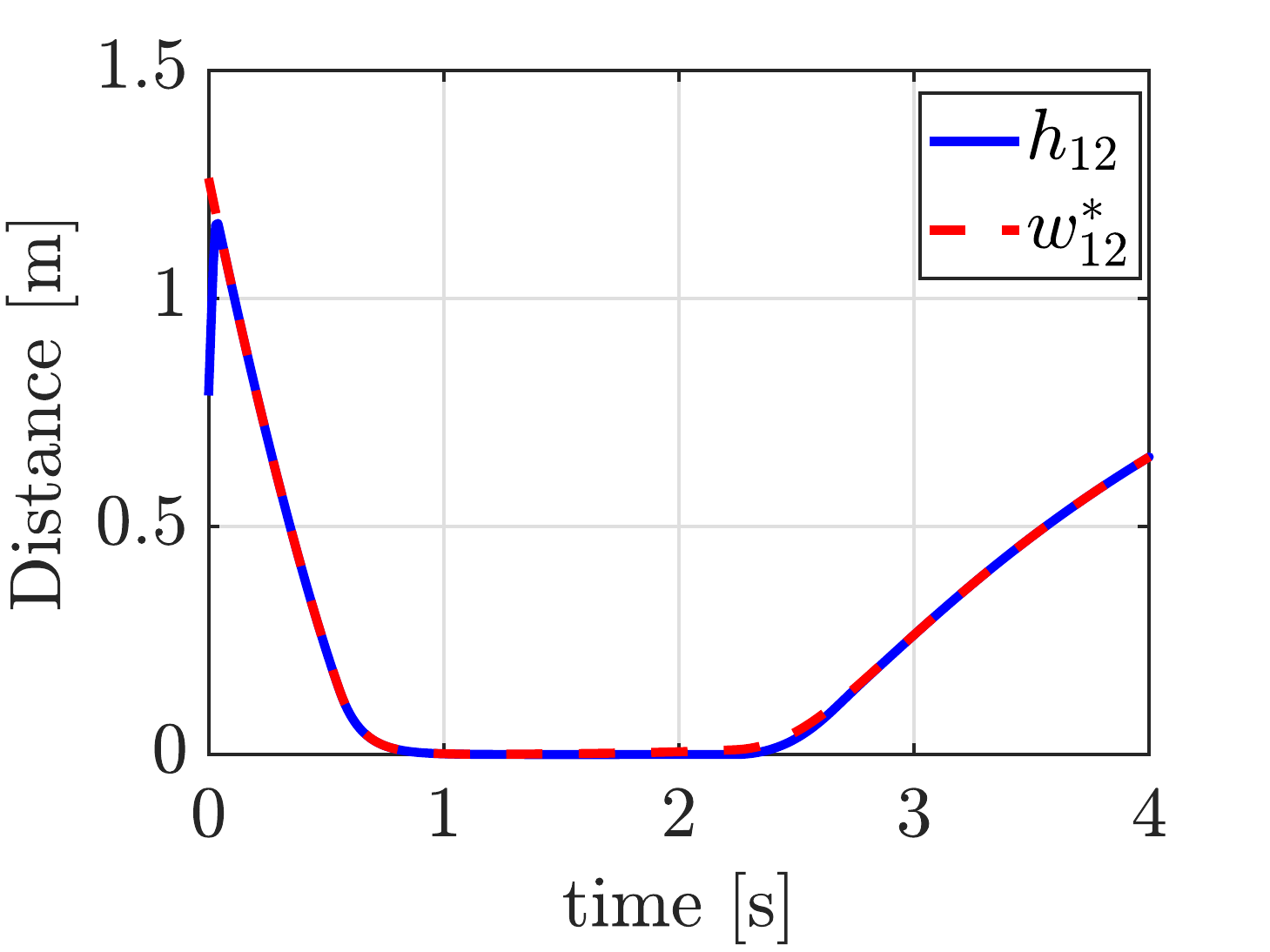}
    \caption{Evolution of $h_{12}(\phi_{12})$, shown in the blue line, and the actual distance $w_{12}^*$ between two elliptical agents, depicted as the red dashed line, calculated by the optimization problem~\eqref{eq:primal}. The control input \eqref{dphi} for $\phi_{12}$ makes $h_{12}(\phi_{12})$ follow $w_{12}^*$ while keeping the smaller value than $w_{12}^*$. Since the value of $h_{12}(\phi_{12})$ remains positive during the simulation, the collision between two elliptical agents is successfully avoided.}
    \label{fig:simu1_h_his}
\end{figure}

\subsection{Simulation with Four Agents}
In the next simulation, we demonstrate our proposed algorithm with the four elliptical agents, the sizes of which are specified by $Q_1=\mathrm{diag}(0.3,0.15)$ for a red agent, $Q_2=\mathrm{diag}(0.4,0.2)$ for a blue agent, $Q_3=\mathrm{diag}(0.4,0.2)$ for a green agent and $Q_4=\mathrm{diag}(0.6,0.3)$ for an orange agent, respectively. 
The initial condition of the simulation is depicted in Fig.~\ref{fig:simu2_state}(a), with the initial pose $x_{1}(0)=\left[-0.1,1.1,-{\pi}/{4}\right]^T$, $x_{2}(0)=\left[1.9,-1.1,-{\pi}/{4}\right]^T$, $x_{3}(0)=\left[-0.1,-1.1,{5\pi}/{4}\right]^T$ and $x_{4}(0)=\left[1.9,1.1,{5\pi}/{4}\right]^T$.
All four agents traverse the environment so that they intersect around the center of the field.
We utilize $\alpha(\bm{h})=10\bm{h}$ for an extended class $\mathcal{K}$ function in CBF conditions and $u_{\mathrm{nom},\phi_{ij}}=20\left({\partial h_{ij}}/{\partial \phi_{ij}}\right)$ for a nominal input to each $\phi_{ij}$.

The snapshots of the simulation are presented in Fig.~\ref{fig:simu2_state}, where each agent is depicted with its trajectory. As illustrated in Fig.~\ref{fig:simu2_state}(b), all agents move straightly toward their goal positions until their distances become closer at the center. Then, the proposed methods modify the nominal input minimally invasive way so that the collision is avoided, as illustrated in Fig.~\ref{fig:simu2_state}(c) and (d). Note that the proposed CBF achieves the collision-free trajectories while changing the attitude of the agents in Fig.~\ref{fig:simu2_state}(c).
Fig.~\ref{fig:simu2_h_his} depicts the transitions of $h_{ij}$, where the proposed methods keep them in the positive value.

\begin{figure}[t!]
    \centering
    \subfloat[$t=0\,{\rm s}$]
    {\includegraphics[trim = 0.75cm 0.1cm 1cm 0.8cm, clip=true, width=0.38\columnwidth]{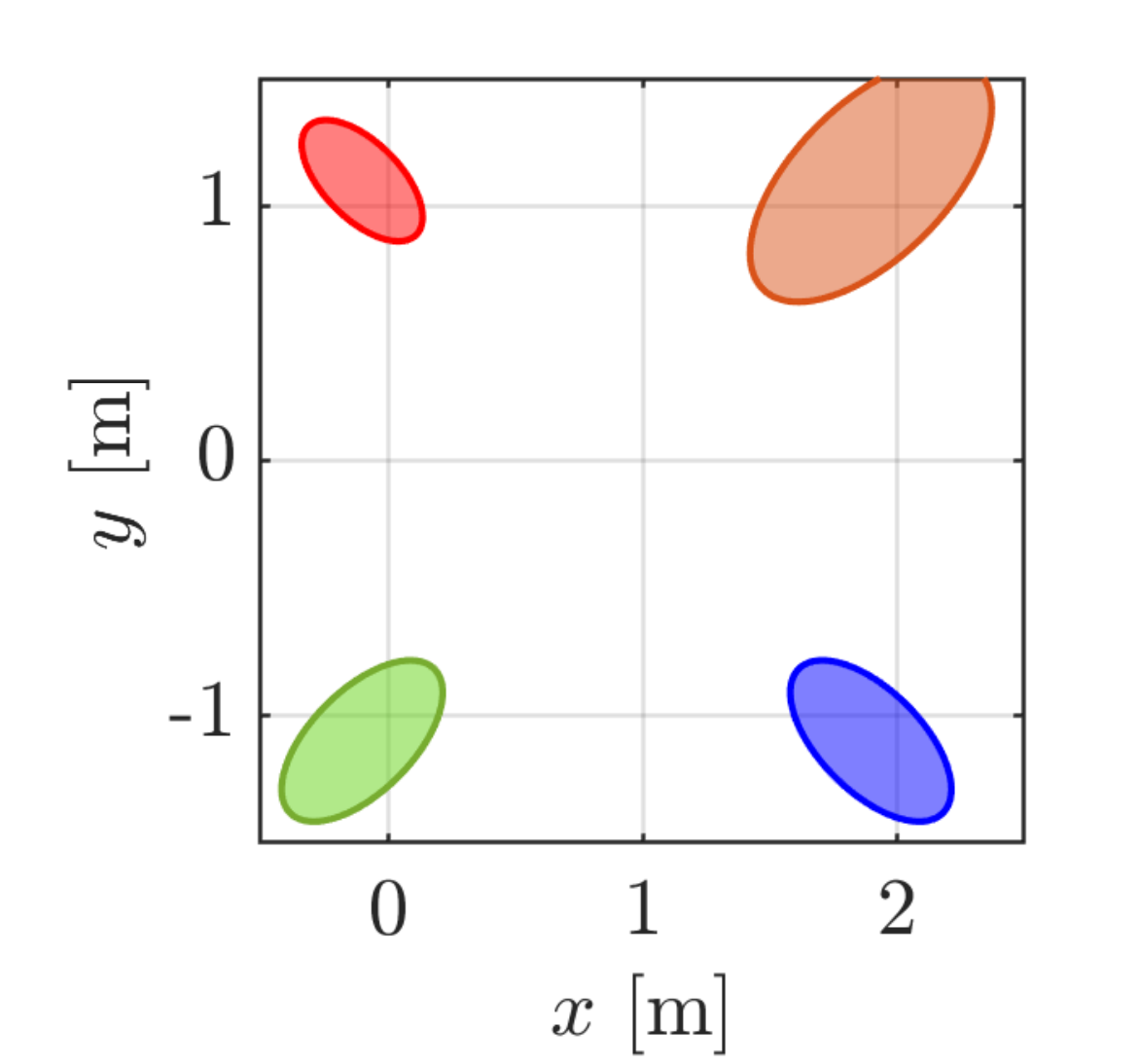}
    \label{subfig:simu2_state_1}} \quad 
    \subfloat[$t=1\,{\rm s}$]
    {\includegraphics[trim = 0.75cm 0.1cm 1cm 0.8cm, clip=true, width=0.38\columnwidth]{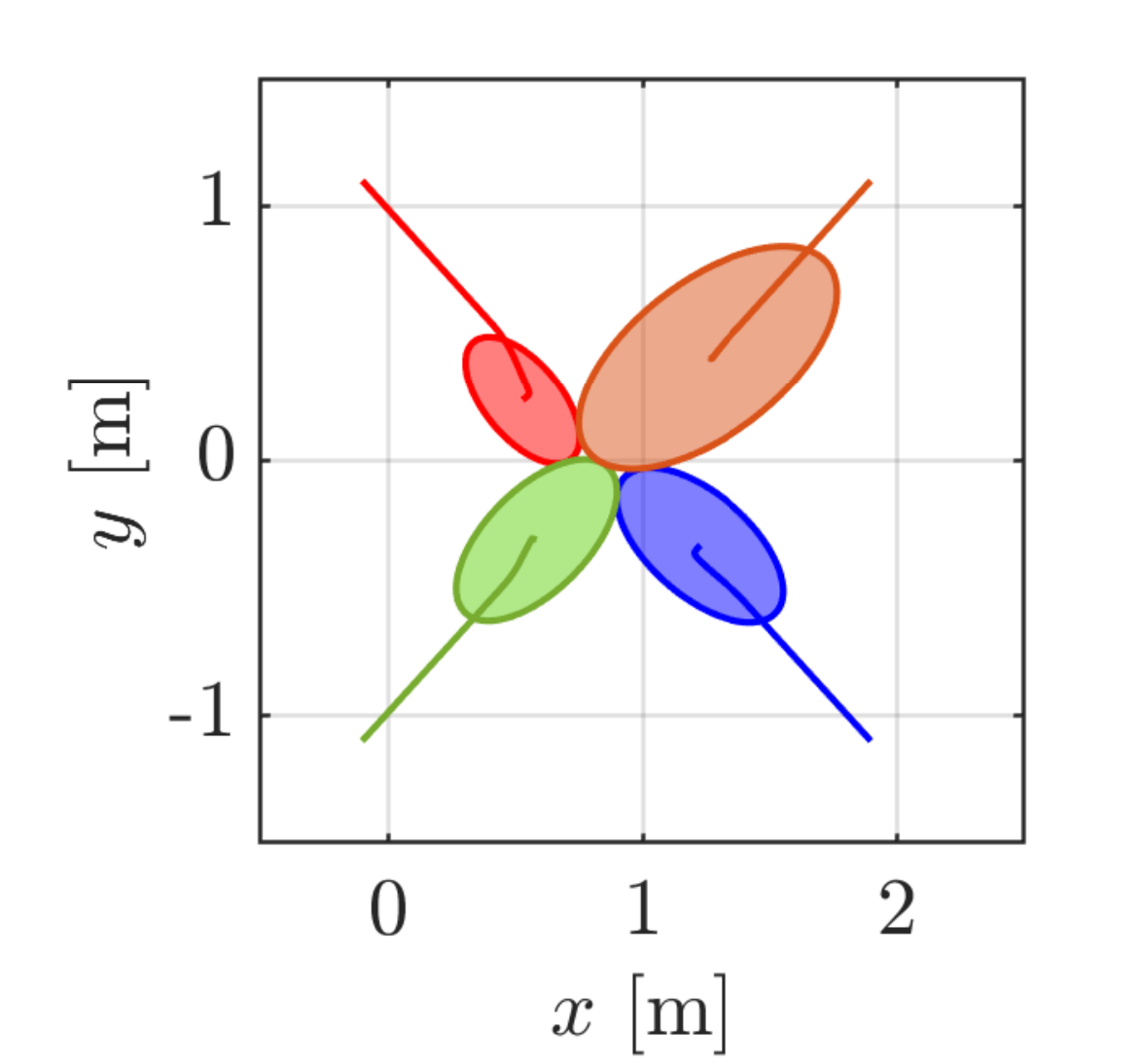}
    \label{subfig:simu2_state_2}} \\ \medskip
    \subfloat[$t=2\,{\rm s}$]
    {\includegraphics[trim = 0.75cm 0.1cm 1cm 0.8cm, clip=true, width=0.38\columnwidth]{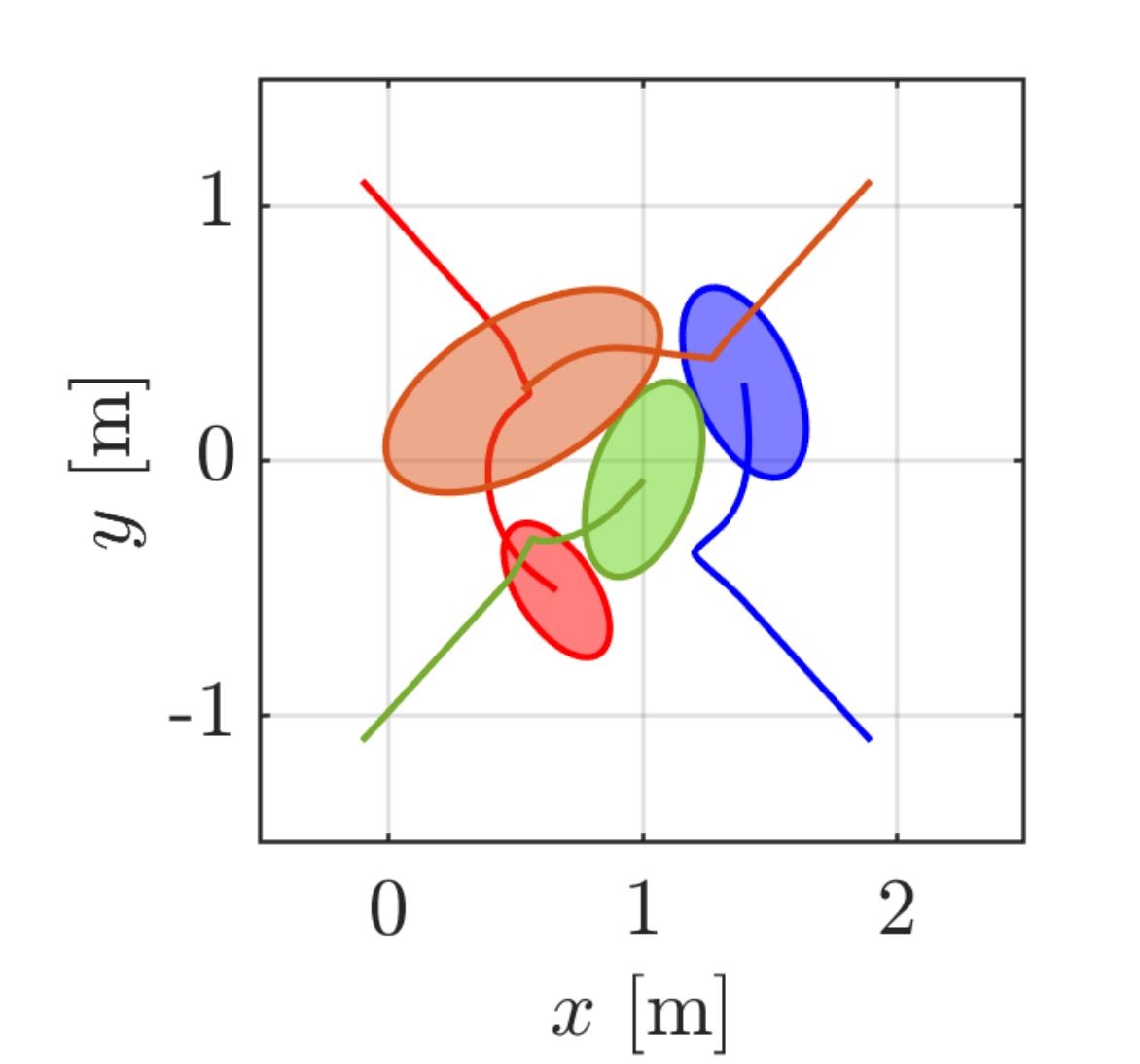}
    \label{subfig:simu2_state_3}} \quad 
    \subfloat[$t=4\,{\rm s}$]
    {\includegraphics[trim = 0.75cm 0.1cm 1cm 0.8cm, clip=true, width=0.38\columnwidth]{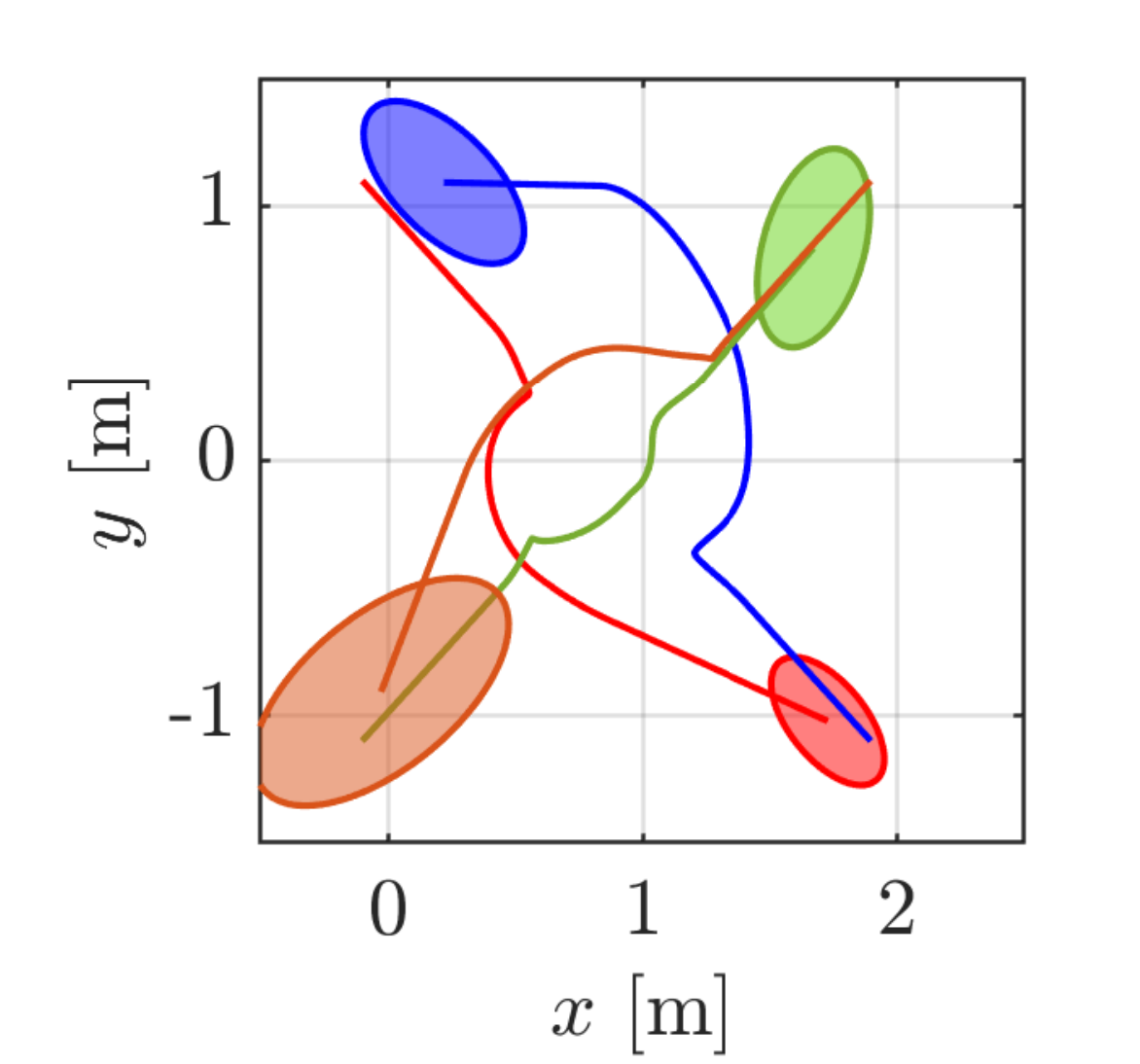}
    \label{subfig:simu_state_4}} 
    \caption{Snapshots of the simulation, where four elliptical agents are depicted together with their trajectories.}
    \label{fig:simu2_state}
\end{figure}

\begin{figure}
    \centering
    \includegraphics[trim = 0cm 0cm 0.5cm 0.3cm, clip=true, width=42mm]{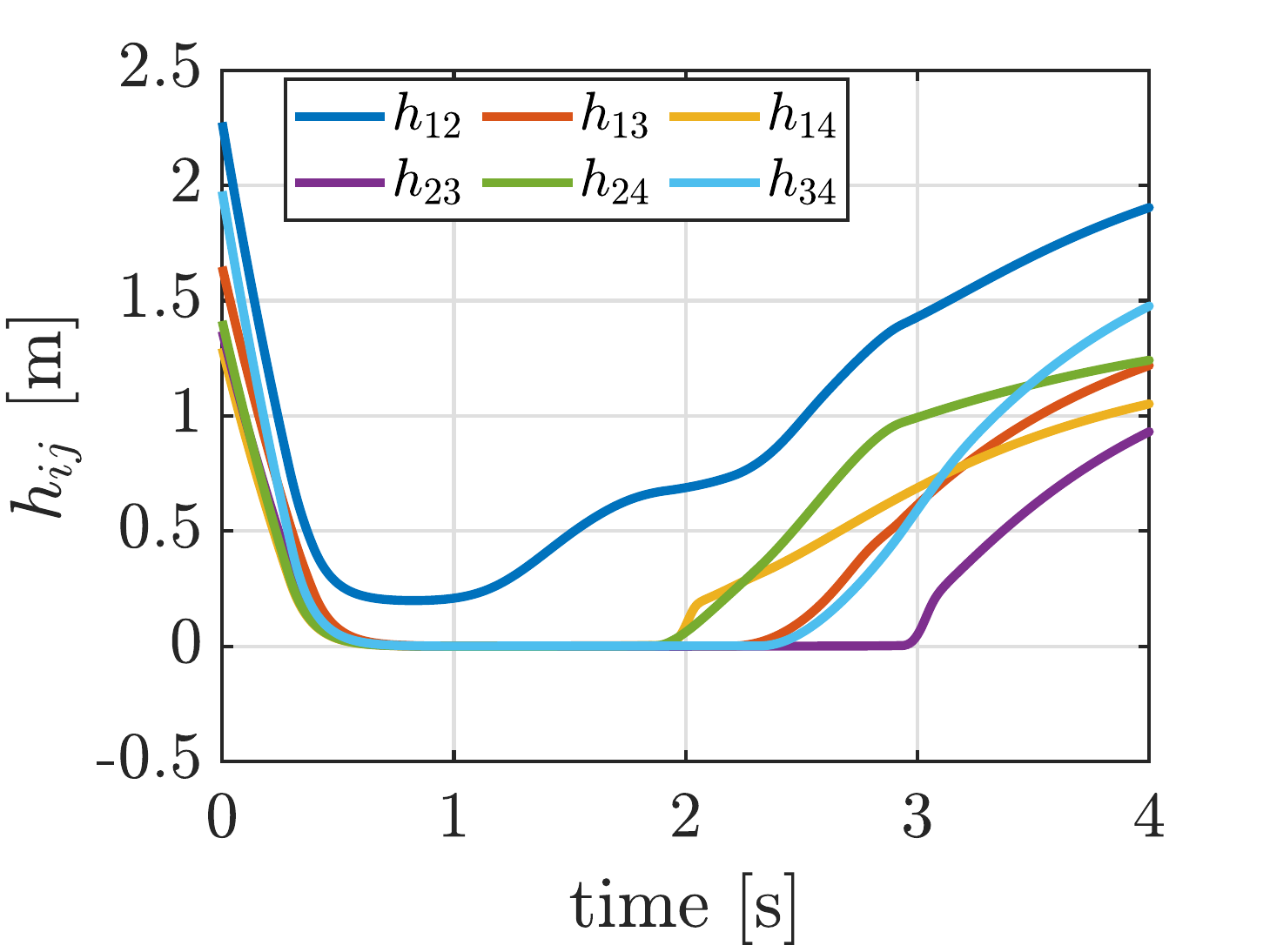}
    \caption{Evolution of $h_{ij},~\forall j \in \mc N\backslash i,~\forall i \in \mc N$. The proposed methods remains $h_{ij}$ in the positive value during the simulation.}
    \label{fig:simu2_h_his}
\end{figure}



\section{Conclusion}

In this paper, we proposed the collision avoidance method for elliptical agents that utilizes the novel CBF leveraging a supporting line between agents. We first introduced a supporting line of an agent to develop the separation condition of agents that is implementable as a CBF. However, we observed that a naive choice of a supporting line might render a shorter distance than the actual distance between two agents. To alleviate the conservativeness in this evaluation, we proposed the optimization problem that rotates the supporting line so that the distance between a supporting line and the other agent is maximized. We then proved that the maximum value derived from this optimization problem is equivalent to the actual distance between two agents. We presented the collision avoidance method incorporating the developed CBF together with the gradient ascent law for rotating the supporting line. Finally, numerical simulations showcased the validity of the proposed methods.
Future works include extending the proposed framework to nonlinear systems through exponential CBF \cite{Koushil2016} while embracing 3D ellipsoidal agents.



\section{Acknowledgement}
We acknowledge Mr. Shunya Yamashita at Tokyo Institute of Technology for helpful discussions on the manuscript.

\bibliographystyle{IEEEtran.bst}
\bibliography{biblio.bib}
\end{document}